\documentclass[smallcondensed]{svjour3} 
\usepackage{lineno,hyperref}
\usepackage{fix-cm}
\usepackage{amsmath}
\usepackage{graphicx}
\usepackage{latexsym,enumerate,amssymb,amsbsy}
\usepackage{graphics,color}
\usepackage{verbatim}
\usepackage{url}
\usepackage{dblfloatfix}
\usepackage{mathptmx}
\usepackage{tikz}
\usepackage{xfrac}
\usepackage{amsfonts}
\usepackage{algorithm}
\usepackage{algpseudocode}
\usepackage{url}
\usepackage{colortbl}

\newcommand{\Mod}[1]{\ (\textup{mod}\ #1)}
\newcommand{\etal}{\emph{et al. }}
\newcommand{\eg}{\emph{e.g.}}
\newcommand{\ie}{\emph{i.e.}}

\def\F{{\mathbb{F}}}
\def\N{{\mathbb{N}}}
\def\Z{{\mathbb{Z}}}
\def\bv{{\mathbf{v}}}
\def\bw{{\mathbf{w}}}
\def\bS{{\mathbf{S}}}
\def\u{{\mathbf{u}}}
\def\s{{\mathbf{s}}}
\def\m{{\mathbf{m}}}
\def\0{{\mathbf{0}}}
\def\1{{\mathbf{1}}}
\def\a{{\mathbf{a}}}
\def\b{{\mathbf{b}}}

\def\x{{\mathbf{x}}}
\def\y{{\mathbf{y}}}
\def\O{{\mathcal{O}}}

\def\M{{\mathcal{M}}}
\def\cC{{\mathcal{C}}}

\def\P{{\mathcal{P}}}
\DeclareMathOperator{\lcm}{lcm}

\def\setcom{{\mathsf{c}}}

\journalname{}

\begin{document}
\title{On Binary de Bruijn Sequences from LFSRs with Arbitrary Characteristic Polynomials}
\author{Zuling Chang \and Martianus Frederic Ezerman \and\\San Ling \and Huaxiong Wang}
\authorrunning{Chang \and Ezerman \and Ling \and Wang}
\institute{Z. Chang \at School of Mathematics and Statistics, Zhengzhou University, Zhengzhou 450001, China\\
\email{zuling\textunderscore chang@zzu.edu.cn}
\and
M. F. Ezerman \and S. Ling \and H. Wang \at Division of Mathematical Sciences, School of Physical and Mathematical Sciences,\\
Nanyang Technological University, 21 Nanyang Link, Singapore 637371\\
\email{\{fredezerman,lingsan,HXWang\}@ntu.edu.sg}
}
\date{Received: date / Accepted: date}
\maketitle

\begin{abstract}
We propose a construction of de Bruijn sequences by the cycle joining method from linear feedback shift registers (LFSRs) with arbitrary characteristic polynomial $f(x)$. We study in detail the cycle structure of the set $\Omega(f(x))$ that contains all sequences produced by a specific LFSR on distinct inputs and provide a fast way to find a state of each cycle. This leads to an efficient algorithm to find all conjugate pairs between any two cycles, yielding the adjacency graph. 
The approach is practical to generate a large class of de Bruijn sequences up to order $n \approx 20$. Many previously proposed constructions of de Bruijn sequences are shown to be special cases of our construction.
\end{abstract}

\section{Introduction}\label{sec:intro}
A binary {\it de Bruijn sequence} of order $n$ has period $N=2^n$
in which each $n$-tuple occurs exactly once in each period. There are $2^{2^{n-1}-n}$ of them~\cite{Bruijn46}. Some of their earliest applications are in communication systems. They are generated in a deterministic way, yet satisfy the randomness criteria in~\cite[Ch.~5]{GG05} and are balanced, containing the same number of $1$s and $0$s. 
In cryptography, they have been used as a source of pseudo-random numbers and in key-sequence generators of stream ciphers~\cite[Sect. 6.3]{Menezes:1996}. In computational molecular biology, one of the three assembly paradigms in DNA sequencing is the de Bruijn graph assemblers model~\cite[Box 2]{NP13}. Some roles of de Bruijn sequences in robust positioning patterns are discussed by Bruckstein \etal in~\cite{Bruck12}. They have numerous applications, \eg, in robotics, smart pens, and camera localization. There has also been an increased interest in deploying de Bruijn sequences in spread spectrum~\cite{SG13}. 

On their construction, two yardsticks are often used to measure the goodness of a method or algorithm, namely, the number of constructed sequences and the efficiency of the construction in terms of both time and memory requirements. In some applications, it is crucial to have a lot of sequences to choose from. Methods to generate all binary de Bruijn sequences have been available in the literature (see, \eg, \cite{Fred82} and \cite{Ral82}). These methods, however, require a large memory space and or long running time.

Fredricksen's survey~\cite{Fred82} recalls various properties and construction methods up to the early $1980$s. A well-known construction called the {\it cycle joining (CJ) method} begins with a given Feedback Shift Register (FSR) and joins all cycles produced by the FSR into a single cycle by identifying the conjugate pairs shared by any pair of cycles. The cycle structure of a Linear FSR (LFSR) can be studied using tools from the algebra of polynomial rings. It is then natural to construct de Bruijn sequences by applying the cycle joining method to LFSRs. Some LFSRs with simple cycle structure, such as the maximal length LFSRs, pure cycling registers, and pure summing registers, have been studied in~\cite{EL84,Fred75,Fred82}. Hauge and Helleseth established a connection between the cycles generated by LFSRs and irreducible cyclic codes in~\cite{HH96}. The number of de Bruijn sequences obtained from these LFSRs is related to cyclotomic numbers, which in general are hard to determine precisely.

Recently, C.~Li \etal studied some classes of de Bruijn sequences. In~\cite{Li14-1} and~\cite{Li14-2}, respectively, the characteristic polynomials of the LFSRs are $(1+x)^3 p(x)$ and $(1+x^3) p(x)$, where $p(x)$ is a primitive polynomial of degree $n>2$. Further generalized results are given in~\cite{Li16} to include products of primitive polynomials whose degrees are pairwise coprime, leading to coprime periods of the sequences that form the cycle structure. This generalization yields a relatively small number of de Bruijn sequences when compared to the one we are proposing. M.~Li and D.~Lin discussed the cycle structure of LFSRs with characteristic polynomial $f(x)=\prod_{k=1}^s \ell_k(x)$ where $\gcd(\ell_i(x),\ell_j(x))=1$ for $1 \leq i \neq j \leq s$ and presented some results about the adjacency graph of $\Omega(f(x))$ in more recent work~\cite{LL17}. Each factor $\ell_k(x)$ of $f(x)$ is not necessarily irreducible.

We put forward a construction from LFSRs whose characteristic polynomials are products of two distinct irreducible polynomials and showed that it generates a large number of de Bruijn sequences in~\cite{Chang2017}. In another work~\cite{CELW16}, whose preliminary results were presented at SETA 2016, we discussed in detail how to determine the cycle structure and find a state for each cycle for an arbitrary
polynomial $f(x) \in \F_{q}[x]$ for any prime power $q$. Drawing insights from them, this present work generalizes the construction of de Bruijn sequences from LFSRs with arbitrary polynomials as their characteristic polynomials. The main contributions are as follows.
\begin{enumerate}
\item We propose a construction of de Bruijn sequences by the cycle joining method from LFSRs with an arbitrary characteristic polynomial in $\F_{2}[x]$. The cycle structure and adjacency graph are studied in details. This gives us a fast method to find all conjugate pairs shared by any two cycles. Our construction covers numerous previously-studied constructions of de Bruijn sequences as special cases. While our approach works for the general case of any arbitrary characteristic polynomial, a careful investigation on LFSRs with repeated roots shows that it is not advisable to construct de Bruijn sequences from them. Unless one needs a lot more de Bruijn sequences than those that can be built from LFSRs whose characteristic polynomials are products of distinct irreducible polynomials. In such a situation one must have access to computing resources beyond what is practical for most users.
	
\item In generalizing the number of irreducible polynomial factors of $f(x)$ from $2$ in~\cite{Chang2017} to $s \geq 2$ we introduce some modifications to stay efficient. Most notably, for an irreducible $g(x)$, a state in $\Omega(g(x))$ is computed based on the mapping $\varphi$ in~\cite[Sect. 3]{Chang2017}, requiring computations over $\F_{2^n}$. This present work determines the state by a simple decimation.
	
\item Based on our theoretical results, a {\tt python} implementation is written to generate de Bruijn sequences. Here we aim for completeness instead of speed. \underline{All} conjugate pairs are found and the \emph{full} adjacency graph $G$ is built. Our method is practical for up to $n \approx 20$. Further optimization tricks are possible, depending on a user's particular requirements and available resources.
\end{enumerate}

After this introduction come preliminary notions and known results in Section \ref{sec:prelims}. Section \ref{sec:structure} discusses the cycle structure and how to find a state belonging to a given cycle. Section \ref{sec:graph} establishes important properties of the conjugate pairs. Two preparatory algorithms are explained in Section \ref{sec:PrepAlg}. These are then used to design the main algorithm that finds all conjugate pairs between any pair of cycles in Section \ref{sec:MainAlg}. Section \ref{sec:reproots} discusses the complication of having a characteristic polynomial with repeated roots and highlights some parts of our method that can still be useful in this situation. Section \ref{sec:generate} shows how the tools fit together nicely, using the examples presented in the preceding sections to derive a large number of de Bruijn sequences of order $7$. A summary of our implementation is given in Section \ref{sec:implem}. The last section contains a brief conclusion, a table listing prior constructions which can be seen as special cases of ours, and possible future directions.

\section{Preliminaries}\label{sec:prelims}

For convenience, we recall needed definitions and results, mostly from~\cite[Ch.~4]{GG05}.

An {\it $n$-stage shift register} is a circuit consisting of $n$ consecutive storage units, each containing a bit, regulated by a clock. When it pulses, the bit in each storage unit is shifted to the next stage in line. A shift register becomes a binary code generator when one adds a feedback loop which outputs a new bit $s_n$ based on the $n$ bits $\s_0= (s_0,\ldots,s_{n-1})$ called an {\it initial state} of the register. The corresponding {\it feedback function} $h(x_0,\ldots,x_{n-1})$ is the Boolean function that outputs $s_n$ on input $\s_0$. A feedback shift register (FSR) outputs a binary sequence $\s=s_0,s_1,\ldots,s_n,\ldots$ satisfying $s_{n+\ell} = h(s_{\ell},s_{\ell+1},\ldots,s_{\ell+n-1})$ for 
$\ell = 0,1,2,\ldots$. For $N \in \N$, if $s_{i+N}=s_i$ for all $i \geq 0$, then $\s$ is {\it $N$-periodic}
or {\it with period $N$} and one writes $\s= (s_0,s_1,s_2,\ldots,s_{N-1})$.  The period of the all zero sequence $\0$ is $1$. When the context is clear, $\0$ also denotes a string of zeroes or a zero vector. We call $\s_i= (s_i,s_{i+1},\ldots,s_{i+n-1})$ {\it the $i$-th state} of $\s$ and states $\s_{i-1}$ and $\s_{i+1}$ the {\it predecessor} and {\it successor} of $\s_i$, respectively. 
%
Given $\a=(a_0,a_1,\ldots,a_{N-1})$ and $\b=(b_0,b_1,\ldots,b_{N-1}) \in \F_{2}^{N}$, let
$c\a :=(ca_0,ca_1,\ldots,ca_{N-1})$ for $c \in \F_2$ and $\a+\b :=(a_0+b_0,a_1+b_1,\ldots,a_{N-1}+b_{N-1})$.

In an FSR, distinct initial states generate distinct sequences, forming the set $\Omega(h)$ of cardinality $2^n$.
All sequences in $\Omega(h)$ are periodic if and only if the feedback function $h$ is {\it nonsingular}, \ie, $h$ can be written as $h(x_0,x_1,\ldots,x_{n-1})=x_0+g(x_1,\ldots,x_{n-1})$, where $g(x_1,\ldots,x_{n-1})$ is some Boolean function with domain $\F_2^{n-1}$~\cite[p.~116]{Golomb81}. In this paper, the feedback functions are all nonsingular. An FSR is called {\it linear} or an LFSR if its feedback function is linear, and {\it nonlinear} or an NLFSR otherwise.

The {\it characteristic polynomial} of an $n$-stage LFSR with feedback $h(x_0,\ldots,x_{n-1})= \sum_{i=0}^{n-1} c_i x_i$ is $f(x)=x^n+\sum_{i=0}^{n-1}c_ix^i \in \F_2[x]$. A sequence $\s$ may have many characteristic polynomials. The one with the lowest degree is the {\it minimal polynomial} of $\s$. It represents the LFSR of shortest length that generates $\s$. Given an LFSR with characteristic polynomial $f(x)$, the set $\Omega(h)$ is also denoted by $\Omega(f(x))$. A sequence $\bv$ is said to be a {\it $d$-decimation} sequence of $\s$, denoted by $\bv=\s^{(d)}$, if $v_j=s_{d \cdot j}$ for all $j \geq 0$. For a sequence $\s$, the {\it (left) shift operator} $L$ is given by
$L\s=L(s_0,s_1,\ldots,s_{N-1})=(s_1,s_2,\ldots,s_{N-1},s_0) $ with the convention that $L^0\s=\s$. The set
$[\s]:=\{\s,L\s,L^2\s,\ldots,L^{N-1}\s \}$ is a {\it shift equivalent class} or a {\it cycle}. The set $\Omega(f(x))$ can be partitioned into cycles. If $\Omega(f(x))$ consists of exactly $r$ cycles $[\s_1],[\s_2],\ldots, [\s_r]$ for some $r \in \N$, then its {\it cycle structure} is $\Omega(f(x))=[\s_1] \cup [\s_2] \cup \ldots \cup [\s_r]$. A {\it conjugate pair} consists of a state $\bv=(v_0,v_{1},\ldots,v_{n-1})$ and its {\it conjugate} $\widehat{\bv}=(v_0 + 1,v_{1},\ldots,v_{n-1})$. Cycles $C_1$ and $C_2$ are {\it adjacent} if they are disjoint and there exists $\bv$ in $C_1$ whose conjugate $\widehat{\bv}$ is in $C_2$. Adjacent cycles with the same feedback $h(x_0,x_1,\ldots,x_{n-1})$ can be joined into a single cycle by interchanging the successors of $\bv$ and $\widehat{\bv}$. The resulting cycle has feedback function
\begin{equation}\label{eq:feedback}
g(x_0,\ldots,x_{n-1})=h(x_0,\ldots,x_{n-1})+\prod_{i=1}^{n-1}(x_i+v_i+1).
\end{equation}
The feedback functions of the resulting de Bruijn sequences are completely determined once the corresponding conjugate pairs are found. It is therefore crucial to find all pairs and to place them in the adjacency graph.

\begin{definition}\cite{HM96}\label{def:adjgraph}
The {\it adjacency graph} $G$ for an FSR with feedback function $h$ is an undirected multigraph whose vertices correspond to the cycles in $\Omega(h)$. There exists an edge between two vertices if and only if they are adjacent. A conjugate pair labels every edge. The number of edges between any pair of cycles is the number of shared conjugate pairs.
\end{definition}

By definition $G$ contains no loops. There is a one-to-one correspondence between the spanning trees of $G$ and the de Bruijn sequences constructed by the
cycle joining method~\cite{HH96,HM96}. The following result, a variant
of the BEST (de {\bf B}ruijn, {\bf E}hrenfest, {\bf S}mith, and {\bf T}utte) Theorem adapted from~\cite[Sect.~7]{AEB87}, provides the counting formula.

\begin{theorem}(BEST)\label{BEST} Let $G$ be the adjacency graph of an FSR with vertex set $\{V_1,V_2,\ldots,V_{\ell}\}$.
Let $\M=(m_{i,j})$ be the $\ell \times \ell$ matrix derived from $G$ in which $m_{i,i}$ is the number of edges incident to vertex $V_i$ and $m_{i,j}$ is the negative of the number of edges between vertices $V_i$ and $V_j$ for $i \neq j$. Then the number of the spanning trees of $G$ is the cofactor of any entry of $\M$.
\end{theorem}
The {\it cofactor} of entry $m_{i,j}$ in $\M$ is $(-1)^{i+j}$ times the determinant of the matrix obtained by deleting the $i$-th row and $j$-th column of $\M$. 
A tool we will need later is the generalized Chinese Remainder Theorem (CRT).

\begin{theorem}(Generalized CRT)~\cite[Thm.~2.4.2]{Ding96}\label{CRT}\\
Let $2 \leq k \in \N$. Given integers $a_1,\ldots,a_k$ and positive integers $m_1,\ldots,m_k$,
there exists $\ell \in \N$ such that 
$\ell \equiv a_i \Mod{m_i} \mbox{ for all } i \in \{1,\ldots,k\}$ if and only if for arbitrary distinct integers $1\leq i\neq j\leq k$, we have $a_i \equiv a_j \Mod{\gcd(m_i,m_j)}$.
	
If $\ell$ is a solution of this system of congruences, then $\ell'$ is also a solution if and only if
$\ell'\equiv\ell\Mod{\lcm(m_1,\ldots,m_k)}$.
\end{theorem}

Let $\{p_1(x), p_2(x), \ldots,p_s(x)\}$ be a set of $s$ pairwise distinct irreducible polynomials
over $\F_2$ and $n_i:=\deg(p_i(x))$. From hereon, let 
\[
f(x):=\prod_{i=1}^{s} p_i(x) \text{ and } n:=\sum_{i=1}^s n_i.
\]

\section{The Cycle Structure and a State in Each Cycle}\label{sec:structure}

Results needed to determine the cycle structure and a state in each cycle have recently been established in~\cite{CELW16}. To make this work self-contained we briefly reproduce the relevant parts in two parts. The first recalls results on the cycle structure of $\Omega(f(x))$. The second determines a state belonging to each of the cycles in $\Omega(f(x))$.

\subsection{The Cycle Structure of $\Omega(f(x))$}\label{subsec:cycle}

Let $g(x)=x^n+b_{n-1}x^{n-1}+\ldots+b_0 \in \F_2[x]$ be an irreducible polynomial of degree $n$ with
a root $\beta \in\F_{2^n}$. Then there exists a primitive element $\alpha\in\F_{2^n}$ such that
$\beta=\alpha^t$ for some $t \in \N$ and $e=\frac{2^n-1}{t}$ is the order of $\beta$. Using the {\it Zech logarithmic representation} (see. {\it e.g.}, \cite[p. 39]{GG05}),
write $1+ \alpha^{\ell} = \alpha^{\tau_{n}(\ell)}$ where $\tau_n (\ell)$ is the Zech
logarithm relative to $\alpha$. It induces a permutation on $\{1,2,\ldots,2^n-2\}$.
Note that $\tau_n (\ell):=\infty$ for $\ell \equiv 0 \Mod{2^n-1}$ and $\alpha^{\infty}:=0$.

The {\it cyclotomic classes} $\cC_i\subseteq\F_{2^n}$, for $0 \leq i <t$, are
\begin{equation}\label{eq:cyclas}
\cC_i=\{\alpha^{i+ s \cdot t}~|~0\leq s <e\}=\{\alpha^i\beta^s~|~0\leq s<e\}=\alpha^i \cC_0.
\end{equation}
The {\it cyclotomic numbers} $(i,j)_{t}$, for $0\leq i,j <t$, are
\begin{equation}\label{eq:cycnum}
(i,j)_{t} =\left|\{\xi~|~\xi\in \cC_i, \xi+1\in \cC_j\}\right|.
\end{equation}
We know from~\cite[Ch.~4]{GG05} that
\begin{equation}\label{equ:g}
\Omega(g(x))=[\0]\cup[\u_0]\cup[\u_1]\cup\ldots\cup[\u_{t-1}]
\end{equation}
with $[\u_i]$ having period $e$. A way to construct $\Omega(g(x))$ with the property that the cyclotomic classes and the cycles are in a one-to-one correspondence can be found in~\cite[Thm. 3]{HH96}. 
We have proved the following result in~\cite{Chang2017} by using the properties of cyclotomic numbers.

\begin{lemma}\label{lem:irr-saa}
Let $g(x) \in \F_2[x]$ be an irreducible polynomial of degree $n$ and order $e$ (making $ t=\frac{2^n-1}{e}$) with $\Omega(g(x))$ as presented in (\ref{equ:g}).
Then, for each triple $(i,j,k)$ with $0 \leq i, j, k < t$, 
\begin{equation}\label{equ:irr-saa}
(j-i,k-i)_t=\left|\{a~|\u_i+L^a\u_j=L^b\u_k;0\leq a,b <e\}\right|.
\end{equation}
\end{lemma}

Given $f(x)$, since $p_i(x)$ is irreducible of degree $n_i$ and order $e_i$, we have $t_i:=\frac{2^{n_i}-1}{e_i}$ and 
\begin{equation}\label{eq:cyc_pi}
\Omega(p_i(x))=[\0] \cup \left[\s^i_0\right] \cup \left[\s^i_1\right]
\cup \ldots \cup \left[\s^i_{t_i-1}\right].
\end{equation}
The cycle structure of $\Omega(f(x))$ can be derived from 
\cite[Thm.~1]{CELW16} by restricting $q$ to $2$. A similar result was established using some properties of LFSRs in~\cite{LL17}. 

\begin{lemma}\label{lem:cycle-f}
Let
$f_i := \displaystyle{
\begin{cases}
e_i & \text{ if } a_i=1\\
1 & \text{ if } a_i=0
\end{cases}}$ and $\delta :=\gcd(f_s,\lcm(f_1,\ldots,f_{s-1}))$. Then 
\begin{equation}\label{eq:shortf}
\Omega(f(x))=\bigcup_{\substack{a_i \in \F_2\\ 1\leq  i \leq s}} ~ \bigcup_{j_1=0}^{t_1-1} \cdots \bigcup_{j_s=0}^{t_s-1}~ \bigcup_{\ell_2=0}^{\gcd(f_2,f_1)-1} \cdots \bigcup_{\ell_{s}=0}^{\delta-1}
\left[a_1\s_{j_1}^{1}+a_2L^{\ell_2}\s_{j_2}^{2}+ \cdots + a_s L^{\ell_{s}} \s_{j_s}^{s}\right].
\end{equation}
\end{lemma}

\subsection{Finding a State belonging to Each Cycle}\label{subsec:state}

Once the number of cycles in $\Omega(f(x))$ is determined, we want to efficiently store them. Recall the state $\s_i$ and its successor $\s_{i+1}$ of an $n$-stage FSR sequence $\s$ with feedback function $h(x_0,\ldots,x_{n-1})$ from Section \ref{sec:prelims}. A {\it state operator} $T_n$ turns $\s_i$ into $\s_{i+1}$ with $s_{i+n} = h(s_i,\ldots,s_{i+n-1})$. The subscript $n$ of $T_{n}$ indicates that $\s_i$ is an $n$-stage state. If $\s_i \in [\s]$ and $e$ is the period of $\s$, then the $e$ distinct states of $[\s]$ are
$\s_i, T \s_i = \s_{i+1},\ldots, T^{e-1} \s_i = \s_{i+e-1}$. It suffices to identify just one state in a given cycle since applying $T$ a suitable number of times generates all $e$ distinct states. To reduce clutters, we use $T$ to denote the state operator for distinct cycles with distinct stages.

Settling this matter is related to~\cite[Problem 2]{LL17}. Deciding if two distinct nonzero states of length $n$ belong to the same cycle is hard to determine if the characteristic polynomial is irreducible but non-primitive. We transform the problem into finding an associated primitive polynomial and use decimation to solve the initial problem.

Let $g(x) \in \F_2[x]$ be an irreducible polynomial of degree $n$, order $e$, and $\beta$ as a root. Hence, $t=\frac{2^n-1}{e}$ and $\Omega(g(x))$ is as given in (\ref{equ:g}). To find a state for each nonzero cycle in $\Omega(g(x))$ we start by searching for a primitive polynomial $q(x)$ of degree $n$ with a root $\alpha$ satisfying $\beta=\alpha^t$. We call $q(x)$ an {\it associated primitive polynomial} of $g(x)$. We know from \cite[Lemma 1]{CELW16} that there are $\displaystyle{\frac{\phi(2^n-1)}{\phi(e)}}$ primitive polynomials that can be associated with $g(x)$ and any one of them can be used here. We use LFSR with characteristic polynomial $q(x)$ and a nonzero initial state to generate an $m$-sequence $\m$. From $\m$, we construct the $t$ distinct $t$-decimation sequences, each of period $e$:
\begin{equation}\label{eq:label}
\u_0=\m^{(t)},\u_1=(L\m)^{(t)},\ldots,\u_{t-1}=(L^{t-1}\m)^{(t)}.
\end{equation} 
These $t$ distinct sequences are in $\Omega(g(x))$.
Starting from an arbitrary $n \cdot t$ consecutive elements of $\m$, one gets $t$ distinct $n$-stage states by decimation. It is then straightforward to verify that each of the derived states corresponds to one nonzero cycle.

To apply Lemma~\ref{lem:irr-saa} later, the correspondence between the cycles and the cyclotomic classes is required. The details can be found in~\cite{HH96}. Hence, one must ensure that $(1,\0)$ is the initial state of $[\u_0]$, corresponding to $\cC_0$. The next proposition shows how to guarantee the correspondence and Algorithm~\ref{algo:state} gives the steps to generate the states. If such correspondence is not necessary, then any nonzero vector in $\F_2^n$ can be used as $\s_0$ in Algorithm~\ref{algo:state}.

\begin{proposition}\label{prop:corresp}
Let a non-primitive irreducible polynomial $g(x)$ and its associated primitive polynomial $q(x)$ be given. Then there exists an initial state $\s_0$ such that $q(x)$ generates an $m$-sequence $\m$ with $(1,\0) \in \F_2^{n}$ as the first $n$ entries in $\m^{(t)}$.
\end{proposition}

\begin{proof}
Using the definition of characteristic polynomial, $\s_0$ can be computed by solving a system of $n$ linear equations. Write $q(x)=x^n + a_{n-1} x^{n-1}+ \ldots + a_1 x + 1$ and let $A$ be its companion matrix: the $n \times n$ matrix whose first row and last column are respectively $(0,\ldots,0,1)$ and $(1,a_1,a_2,\ldots,a_{n-1})^{\intercal}$. The remaining entries form the identity matrix $I_{n-1}$. Then the respective first entry of the state vectors $\s_0, \s_0 A^{t}, \s_0 A^{2t},\ldots,\s_0 A^{(n-1)t}$ must be $1,0,0,\ldots,0$. Solving the system gives us $\s_0$. \qed
\end{proof}

\begin{algorithm}[h!]
\caption{Finding a State in a Nonzero Cycle in $\Omega(g(x))$}
\label{algo:state}
\begin{algorithmic}[1]
\renewcommand{\algorithmicrequire}{\textbf{Input:}}
\renewcommand{\algorithmicensure}{\textbf{Output:}}
\Require An irreducible polynomial $g(x) \in \F_2[x]$.
\Ensure A state $\bv_j$ of each nonzero cycle $[\u_j] \in \Omega(g(x))$.
\State $e \gets$ order of $g(x)$; $t \gets (2^n-1)/e$.
\If {$t=1$}
	\State $\bv_0 \gets (1,0,\ldots,0)\in\F_2^n$ and break.
\Else
	\State $q(x) \gets$ an associated primitive polynomial of $g(x)$.\label{primpoly}
	\State $\bw\gets$ the first $n \cdot t$ consecutive entries of the $m$-sequence generated by $q(x)$ on input $\s_0 \in \F_2^n$.
	\For {$j$ from $0$ to $t-1$}
		\State $\bv_j \gets (L^{j} \bw)^{(t)}$.
	\EndFor
\EndIf
\end{algorithmic}
\end{algorithm}

\begin{example}\label{example1}
Consider $g(x)=x^4+x^3+x^2+x+1$, an irreducible polynomial of order $5$, with $p(x)=x^4+x+1$ as the associated primitive polynomial. The $m$-sequence that $p(x)$ generates on input $(1,0,0,0)$ is $\m=(1,0,0,0,1,0,0,1,1,0,1,0,1,1,1,1,0,0,0,1,0,0,1,1,0)$. Computing $(L^{j} \m)^{(3)}$ with $j \in \{0,1,2\}$ gives us $[\u_j]$. We take as the respective initial states the first $4$ entries of $(L^{j} \m)^{(3)}$: $(1,0,0,0)$, $(0,1,1,1)$, and $(0,0,1,0)$.
\end{example}

Based on known states of the cycles in $\Omega(p_i(x))$, we 
determine a state in each of the cycles in $\Omega(f(x))$. For each $1 \leq i \leq s$, construct the $n_i \times n$ matrix $\P_i$ in the following manner.
The $j$-th row of $\P_i$ is the first $n$ bits of the sequence generated by the LFSR with
characteristic polynomial $p_i(x)$ whose $n_i$-stage initial state has $1$ in the $j$-th position and $0$ elsewhere.
We combine the resulting matrices into the full-rank (see~\cite[Lemma 7]{CELW16}) matrix
\[
\P_{n \times n}=
\begin{pmatrix}
\P_1\\
\P_2\\
\vdots\\
\P_s
\end{pmatrix}.
\]

Let $\P$ be already constructed. Let $\bv \in \F_2^{n}$ and $\a_i \in \F_2^{n_i}$, with $1\leq i\leq s$, be respectively
the $n$-stage and $n_i$-stage states of the sequences in $\Omega(f(x))$ and $\Omega(p_i(x))$. There is a bijection between $\bv$ and $(\a_1,\a_2,\ldots,\a_s)$ via $\bv=(\a_1,\a_2,\ldots,\a_s) \P$. Note that $\P$ and $T$ commute since $T \bv =
T[(\a_1,\a_2,\ldots,\a_s) \P]=(T\a_1,T\a_2,\ldots,T\a_s) \P $.
Hence, any sequence $\s \in \Omega(f(x))$ with initial state $\bv$ can be written as the sum of
sequences $\s_i$ from $\Omega(p_i(x))$ with corresponding initial states $\a_i$ for all $i$.
We can conveniently use $(\a_1,\a_2,\ldots,\a_s)$ to represent the state $\bv$. A sequence $\s$ generated by $f(x)$ has a longer period and is more complex to study than the {\it component sequences} $\s_i$.
Representing the state $\bv$ of $\s$ in terms of states
$\a_i$ of corresponding sequences $\s_i$ helps us study the global properties while relying on local properties only. The new state representation offers a significant gain in storage efficiency.
In addition, a state belonging to each cycle in $\Omega(f(x))$ can be quickly computed using the representation.

Suppose that we have obtained the set $A_i$ of states corresponding to
the $t_i +1$ distinct cycles in $\Omega(p_i(x))$ given in (\ref{eq:cyc_pi}). Letting
$\0$ be the state of $[\0]$ and $\a^i_j$ a nonzero state of $[\s^i_j]$,
$
A_i:=\{\a^i_0,\a^i_1,\ldots,\a^i_{t_i-1},\a^i_{t_i}=\0\}$. For convenience, let each of the states be the initial state of its corresponding sequence. Then one takes
\begin{equation}\label{eq:state}
\bv=\left(\a^1_{j_1},\a^2_{j_2},\ldots,\a^s_{j_s}\right) \P \text{ with } \a^i_{j_i}\in A_i
\end{equation}
as an initial state of a sequence $\s \in \Omega(f(x))$. Note that $\s$ has the form $a_1\s^1_{j_1}+a_2\s^2_{j_2}+\ldots+a_s\s^s_{j_s}$, where $a_i \s^i_{j_i}=\0$ if $a_i=0$. For all other cases, $a_i=1$. For $1 \leq i \leq s$, let $\ell_i$ be a nonnegative integer and let $\bv$ be as
given in (\ref{eq:state}). By the properties of $\P$ and $T$, $\bw=\left(T^{\ell_1}\a^1_{j_1},T^{\ell_2}\a^2_{j_2},\ldots,T^{\ell_s}\a^s_{j_s}\right) \P$ is a state of cycle
$
\left[a_1L^{\ell_1}\s^1_{j_1}+a_2L^{\ell_2}\s^2_{j_2}+\ldots+a_sL^{\ell_s}\s^s_{j_s}\right]=
\left[a_1\s^1_{j_1}+a_2L^{\ell_2-\ell_1}\s^2_{j_2}+\ldots+a_sL^{\ell_s-\ell_1}\s^s_{j_s}\right]$. We quickly find a state belonging to any cycle. 

\begin{example}\label{example2}
Let $\displaystyle{f(x)=\underbrace{(x+1)}_{=p_1(x)} \underbrace{(x^2+x+1)}_{=p_2(x)} \underbrace{(x^4+x^3+x^2+x+1)}_{=p_3(x)}}$. Notice that $p_3(x)$ is not primitive. One gets $\P_1=[\1]$,
$\P_2=
\begin{pmatrix}
1 & 0 & 1 & 1 & 0 & 1 & 1  \\
0 & 1 & 1 & 0 & 1 & 1 & 0
\end{pmatrix}$,
$\P_3=
\begin{pmatrix}
1 & 0 & 0 & 0 & 1 & 1 & 0  \\
0 & 1 & 0 & 0 & 1 & 0 & 1  \\
0 & 0 & 1 & 0 & 1 & 0 & 0  \\
0 & 0 & 0 & 1 & 1 & 0 & 0
\end{pmatrix}$, from which $\P$ and $\P^{-1}$ immediately follow.
The relevant cycles and sets of states are
\begin{align*}
\Omega(p_1(x))=&[\0]\cup[\s^1_0=\1], A_1=\{\a^1_0=(1),\a^1_1=(0)\},\\
\Omega(p_2(x))=&[\0]\cup[\s^2_0=(1,0,1)], A_2=\{\a^2_0=(1,0),\a^2_1=(0,0)\},\\
\Omega(p_3(x))=&[\0]\cup[\s^3_0=(1,0,0,0,1)]\cup[\s^3_1=(0,1,1,1,1)]
\cup[\s^3_2=(0,0,1,0,1)],\\
A_3=&\{\a^3_0=(1,0,0,0),\a^3_1=(0,1,1,1),\a^3_2=(0,0,1,0),\a^3_3=(0,0,0,0)\}.
\end{align*}
The periods of the nonzero sequences in $\Omega(p_i(x))$ are $1,3$, and $5$. We write each cycle in $\Omega(f(x))$ as
$\left[a_1\1+a_2\s^2_0+a_3\s^3_j\right]$ with $a_i \in \F_2$ and $j \in \{0,1,2\}$. Choosing $\a=\left(\a^1_0,\a^2_0,\a^3_0\right)=(1,1,0,1,0,0,0)$ implies that $\bv= \a \P=(1,1,0,0,0,1,0)$ is a state of $\left[\1+\s^2_0+\s^3_0\right]$. A state of each of the cycles in $\Omega(f(x))$ can be similarly derived. Table~\ref{table:ex} lists them down.
	

\begin{table}[h!]
\caption{List of States and Respective Cycles}
\label{table:ex}
\renewcommand{\arraystretch}{1.3}
\centering
\begin{tabular}{cccc}
\hline
Apply $\P$ to & State  & Cycle & Period \\
\hline
$\left(\a^1_1,\a^2_1,\a^3_3\right)$ & $\0$ & $[\0]$ & $1$ \\
$\left(\a^1_1,\a^2_0,\a^3_3\right)$ & $(1,0,1,1,0,1,1)$ & $\left[\s_0^{2}\right]=[(1,0,1)]$ & $3$ \\
			
$\left(\a^1_1,\a^2_1,\a^3_0\right)$ & $(1,0,0,0,1,1,0)$ & $\left[\s_0^{3}\right]=[(1,0,0,0,1)]$ & $5$\\
			
$\left(\a^1_1,\a^2_1,\a^3_1\right)$ & $(0,1,1,1,1,0,1)$ & $\left[\s_1^{3}\right]=[(0,1,1,1,1)]$ & $5$\\
			
$\left(\a^1_1,\a^2_1,\a^3_2\right)$ & $(0,0,1,0,1,0,0)$ & $\left[\s_2^{3}\right]=[(0,0,1,0,1)]$ & $5$\\
			
$\left(\a^1_1,\a^2_0,\a^3_0\right)$ & $(0,0,1,1,1,0,1)$ & $\left[\s^2_0+\s^3_0\right]
=[(0,0,1,1,1,0,1,0,1,0,1,1,1,0,0)]$ & $15$ \\
			
$\left(\a^1_1,\a^2_0,\a^3_1\right)$ & $(1,1,0,0,1,1,0)$ & $\left[\s^2_0+\s^3_1\right]
=[(1,1,0,0,1, 1,0,1,0,0, 0,0,0,1,0)]$ & $15$ \\
			
$\left(\a^1_1,\a^2_0,\a^3_2\right)$ & $(1,0,0,1,1,1,1)$ & $\left[\s^2_0+\s^3_2\right]
=[(1,0,0,1,1,1,1,1,1,0,0,1,0,0,0)]$ & $15$ \\
			
$\left(\a^1_0,\a^2_1,\a^3_3\right)$ & $\1$ & $[\1]$ & $1$ \\
$\left(\a^1_0,\a^2_0,\a^3_3\right)$ & $(0,1,0,0,1,0,0)$ & $\left[\1+\s_0^{2}\right]=[(0,1,0)]$ & $3$\\
			
$\left(\a^1_0,\a^2_1,\a^3_0\right)$ & $(0,1,1,1,0,0,1)$ & $\left[\1+\s_0^{3}\right]=[(0,1,1,1,0)]$ & $5$ \\
			
$\left(\a^1_0,\a^2_1,\a^3_1\right)$ & $(1,0,0,0,0,1,0)$ & $\left[\1+\s_1^{3}\right]=[(1,0,0,0,0)]$ & $5$ \\
			
$\left(\a^1_0,\a^2_1,\a^3_2\right)$ & $(1,1,0,1,0,1,1)$ & $\left[\1+\s_2^{3}\right]=[(1,1,0,1,0)]$ & $5$ \\
			
$\left(\a^1_0,\a^2_0,\a^3_0\right)$ & $(1,1,0,0,0,1,0)$ & $\left[\1+\s^2_0+\s^3_0\right]
=[(1,1,0,0,0,1,0,1,0,1,0,0,0,1,1)]$ & $15$ \\
			
$\left(\a^1_0,\a^2_0,\a^3_1\right)$ & $(0,0,1,1,0,0,1)$ & $\left[\1+\s^2_0+\s^3_1\right]
=[(0,0,1,1,0, 0,1,0,1,1, 1,1,1,0,1)]$ & $15$ \\
			
$\left(\a^1_0,\a^2_0,\a^3_2\right)$ & $(0,1,1,0,0,0,0)$ & $\left[\1+\s^2_0+\s^3_2\right]
=[(0,1,1,0,0, 0,0,0,0,1, 1,0,1,1,1)]$ & $15$ \\
\hline
\end{tabular}
\end{table}
\end{example}

Example~\ref{example2} demonstrates how our approach quickly finds a state belonging to any cycle in $\Omega(f(x))$. This works in general, not only when the periods are coprime, as will be shown in Section~\ref{sec:reproots} below. 

\section{Properties of the Adjacency Graph of $\Omega(f(x))$}\label{sec:graph}
To use Lemma~\ref{lem:irr-saa}, we assume that each nonzero cycle $[\s^i_j] \in \Omega(p_i(x))$ in (\ref{eq:cyc_pi}) corresponds 
to the suitable cyclotomic class. No generality is loss here. If the two do not correspond, the resulting adjacency graph would be permutation equivalent to the one we obtain.

Let $C_1=[\s_1]$ and $C_2=[\s_2]$ be two cycles in $\Omega(f(x))$ and let the special state
$\bS:=(1,\0) \in \F_2^{n}$ belong to $[\s_0] \in \Omega(f(x))$.
If $(\bv,\widehat{\bv})$ is a conjugate pair between $C_1$ and $C_2$, then $\widehat{\bv}+\bv=\bS$.
Thus, $C_1$ and $C_2$ share at least a conjugate pair if and only if
some shift of $\s_1$ plus some shift of $\s_2$ is equal to $\s_0$, \ie,
there exist integers $\ell$ and $\ell'$ satisfying $L^{\ell}\s_1+L^{\ell'}\s_2=\s_0$. Since $\bS $ has $n-1$ consecutive $0$s, $f(x)$ is the minimal polynomial of $[\s_0]$. Hence, $[\s_0]$ can be determined with the
help of $\P^{-1}$ and Equation (\ref{eq:state}). Without loss of generality, let
\begin{equation}\label{eq:cons}
[\s_0]=\left[L^{c_1}\s_{d_1}^{1}+L^{c_2}\s_{d_2}^{2}+ \ldots + L^{c_{s}}\s_{d_s}^s\right]
\end{equation}
where $c_1,c_2,\ldots,c_{s}$ and $d_1,d_2,\ldots, d_{s}$ are some suitable integers.
In fact, if the initial states are appropriately chosen, one can make $c_k=d_k=0$ for
$1 \leq k \leq s$. For our purposes, doing so is not required. 

It is clear that $[\0]$ and $[\s_0]$ share a unique conjugate pair. To compute the exact number of conjugate pairs between $C_1$ and $C_2$,
knowing important properties of the cycles in $\Omega(p_i(x))$ is crucial. Suppose that we have \begin{equation}\label{ss:tuple-p_i}
\Gamma_{i}(j_1,j_2):=\{(u,v)|L^u\s^i_{j_1}+L^v\s^i_{j_2}=\s_{d_i}^i,~ 0\leq u, v<e_i\}.
\end{equation}
If $\s^i_{j_1}=\0$, then $u=0$. If $\s^i_{j_2}=\0$, then $v=0$.
If only one of $\s^i_{j_1}$ and $\s^i_{j_2}$ is $\0$, then the other must be $\s_{d_i}^i$
and $\Gamma_{i}(j_1,j_2)=\{(0,0)\}$.
If both $\s^i_{j_1}$ and $\s^i_{j_2}$ are $\neq \0$, then
$|\Gamma_{i}(j_1,j_2)|=(j_1-d_i,j_2-d_i)_{t_i}=(j_2-d_i,j_1-d_i)_{t_i}$ defined in Lemma \ref{lem:irr-saa}. Following Equation (\ref{eq:shortf}), we let
\begin{equation}\label{equ:cc}
C_1=[\s_1]=\left[a_1L^{\ell_1}\s_{j_1}^{1}+
\ldots+a_sL^{\ell_{s}}\s_{j_s}^{s}\right] \mbox{ and } C_2=[\s_2]=\left[a'_1L^{\ell_1'}\s_{j_1'}^{1}+
\ldots+a_s'L^{\ell_{s}'}\s_{j_s'}^{s}\right].
\end{equation}
Both $\s_1$ and $\s_2$ have period $\lcm(f_1,\ldots,f_s)$. Let
$E_1:=\{i~|~a_i=1, 1\leq i\leq s\}$ and $E_2:=\{j~|~a_j'=1, 1 \leq j \leq s\}$.

\begin{theorem}\label{thm:cp}
Take $C_1$ and $C_2$ from (\ref{equ:cc}) and $\Gamma_{i}(j_i,j_i')$ in (\ref{ss:tuple-p_i}).
\begin{enumerate}
\item The following conditions are necessary for $C_1$ and $C_2$ to share at least a conjugate pair.
	\begin{enumerate}
	\item $E_1\cup E_2=\{1,2,\ldots,s\}$.
	\item If $i \in E_1 \cap {E_2}^{\setcom}$, then $\s_{j_i}^{i}=\s_{d_i}^{i}$. Similarly, if $i \in {E_1}^{\setcom} \cap E_2$, then $\s_{j_i'}^{i}=\s_{d_i}^{i}$.
	\item If $i \in E_1 \cap E_2$, then $(j_{i}-d_i,j'_{i}-d_i)_{t_i}>0$.
\end{enumerate}
\item The number of conjugate pairs shared by cycles $C_1$ and $C_2$ is equal to the number of tuples $((u_1,v_1),\ldots,(u_s,v_s))$ that satisfy two requirements.
\begin{enumerate}
	\item $(u_i,v_i)\in \Gamma_{i}(j_i,j_i')$ for $1 \leq i \leq s$.
	\item The following systems of congruences hold modulo $\gcd(e_i,e_j)$.
	\begin{align*}
	&c_i+u_i-\ell_i\equiv  c_j+u_j-\ell_j  \mbox{ for all }i\neq j\in E_1,\\
	&c_i+v_i-\ell_i'\equiv c_j+v_j-\ell_j' \mbox{ for all }i\neq j\in E_2.
	\end{align*}
	\end{enumerate}
	\item The sum of the numbers of conjugate pairs between any two cycles over all possible values for $\ell_1, \ldots, \ell_s$ and $\ell_1', \ldots, \ell_s'$ is equal to $\displaystyle \prod_{i \in E_1 \cap E_2} (j_{i}-d_i,j_{i}'-d_i)_{t_i}$.
		
	As the $\ell$s and $\ell'$s range through all of their respective values, it may happen that $C_1=C_2$. When this is the case, we count the conjugate pairs $(\bv,\widehat{\bv})$ and $(\widehat{\bv},\bv)$ separately even though they are the same.
\end{enumerate}
\end{theorem}

\begin{proof}
If $C_1$ and $C_2$ share a conjugate pair, then there exist integers $\ell$ and $\ell'$ satisfying $L^{\ell}\s_1+L^{\ell'}\s_2=\s_0$. By definitions,
\[
L^{\ell}\s_1+L^{\ell'}\s_2=\sum_{i=1}^{s}a_iL^{\ell+\ell_i}\s_{j_i}^{i}+a'_iL^{\ell'+\ell_i'}\s_{j_i'}^{i}= \\
\sum_{i=1}^{s}L^{c_i} \left[a_iL^{\ell+\ell_i-c_i}\s_{j_i}^{i}+a'_iL^{\ell'+\ell_i'-c_i}\s_{j_i'}^{i}\right]
=\sum_{i=1}^{s}L^{c_i}\s_{d_i}^{i}.
\]
One gets $a_iL^{\ell+\ell_i-c_i}\s_{j_i}^{i}+a'_iL^{\ell'+\ell_i'-c_i}\s_{j_i'}^{i}=\s_{d_i}^{i}$ since sequences $\s_{d_i}^{i}$ correspond to distinct irreducible characteristic polynomials. Therefore, $a_i \vee a_i'=1$, $(\ell+\ell_i-c_i, \ell'+\ell_i'-c_i)\in \Gamma_i(j_i,j_i')$, and $|\Gamma_i(j_i,j_i')|>0$, proving Statement 1.
	
Let $(u_i,v_i)$ be one of the $(j_{i}-d_i,j_{i}'-d_i)_{t_i}$ tuples in $\Gamma_i(j_{i},j_{i}')$.
To ensure the existence of $\ell$ and $\ell'$, the following systems of congruences must have solutions modulo $f_i$ for all $i$:
\[
\begin{cases}
u_i \equiv \ell+\ell_i-c_i \\
v_i\equiv\ell'+\ell_i'-c_i
\end{cases} \iff \quad 
\begin{cases}
\ell\equiv c_i+u_i-\ell_i \\
\ell'\equiv c_i+v_i-\ell_i'
\end{cases}.
\]
By the definition of $f_i$, with $u_i=v_i=0$ whenever $i \notin E_1 \cap E_2$, the systems reduce to
\[
\begin{cases}
\ell \equiv c_i+u_i-\ell_i \Mod{e_i}  \mbox{ for all } i \in E_1\\
\ell'\equiv c_i+v_i-\ell_i' \Mod{e_i} \mbox{ for all } i \in E_2
\end{cases}.
\]
	
Theorem \ref{CRT} says that solutions exist if and only if the following systems of congruences hold modulo $\gcd(e_i,e_j)$:
\[
\begin{cases}
c_i+u_i-\ell_i \equiv  c_j+u_j-\ell_j \mbox{ for all } i\neq j\in E_1\\
c_i+v_i-\ell_i'\equiv c_j+v_j-\ell_j' \mbox{ for all } i\neq j\in E_2
\end{cases}.
\]
This completes the proof of Statement 2.
	
As $\ell_i$ and $\ell_i'$ range through their possible values, for a given $(u_i,v_i)$ there exists a corresponding $(\ell, \ell')$ satisfying the above systems simultaneously. There are $(j_{i}-d_i,j_{i}'-d_i)_{t_i}$ possible choices for $(u_i,v_i)$ if both corresponding sequences are $\neq \0$ but only one choice if one of them is $\0$. Thus, the sum of the numbers of conjugate pairs between $C_1$ and $C_2$ over all possible $\ell_i$s and $\ell_i'$s is $\displaystyle{\prod_{i \in E_1 \cap E_2} (j_{i}-d_i,j_{i}'-d_i)_{t_i}}$. When $C_1=C_2$, we double count each of their conjugate pairs since both $(\bv,\widehat{\bv})$ and $(\widehat{\bv},\bv)$ appear and are counted separately, despite being exactly the same. Statement 3 is established. \qed
\end{proof}

\begin{remark}\label{rem:Liresult}
A similar result to Theorem~\ref{thm:cp} was given in~\cite[Thm.~4]{LL17}. We provide the proof above using our method.
\end{remark}

To determine the number of conjugate pairs between two chosen cycles, for each possible $(\ell_i,\ell_i')$ tuple, we check how many $((u_1,v_1),\ldots,(u_s,v_s))$ tuples make the systems solvable. We will provide an efficient procedure to do so. Theorem \ref{thm:cp} says that in general it is hard to determine
the exact number of conjugate pairs between any two cycles in $\Omega(f(x))$.
In some special cases the problem becomes easier. If, for all $1 \leq i \leq s$, the periods of the nonzero sequences in $\Omega(p_i(x))$
are pairwise coprime, then, by Lemma \ref{lem:cycle-f}, $[\s_0] =\left[\s_{d_1}^{1} + \s_{d_2}^{2} + \ldots + \s_{d_s}^s\right]$,
\begin{equation}\label{eq:copr}
C_1=[\s_1] =\left[a_1\s_{j_1}^{1} + \ldots+ a_s\s_{j_s}^{s}\right]\mbox{, and }
C_2=[\s_2] =\left[a'_1\s_{j_1'}^{1} + \ldots+ a_s'\s_{j_s'}^{s}\right].
\end{equation}
We obtain the following corollary to Theorem~\ref{thm:cp}.

\begin{corollary}\label{cor:1}
For all $1 \leq i \leq s$ let the periods of the nonzero sequences in $\Omega(p_i(x))$
be pairwise coprime. Let $C_1$ and $C_2$ be as given in (\ref{eq:copr}). They share at least a conjugate pair
if and only if the following requirements are satisfied for all $i$.
\begin{enumerate}
\item $E_1\cup E_2=\{1,2,\ldots,s\}$.
\item If $a_i=1$ and $a_i'=0$, then $\s_{j_i}^{i}=\s_{d_i}^{i}$.
\item If $a_i=0$ and $a_i'=1$, then $\s_{j_i'}^{i}=\s_{d_i}^{i}$.
\item If $a_i=a_i'=1$, then $(j_{i}-d_i,j_{i}'-d_i)_{t_i}>0$.
\end{enumerate}
	
Distinct $C_1$ and $C_2$ share $\displaystyle{\prod_{i \in E_1 \cap E_2}(j_{i}-d_i,j_{i}'-d_i)_{t_i}}$
conjugate pairs. We halve the number when $C_1=C_2$.
\end{corollary}

If $p_i(x)$ in Corollary~\ref{cor:1} are all primitive, then $t_i=1$ for all $i$ and
there is a conjugate pair between $C_1$ and $C_2$ in (\ref{eq:copr}) if and only if
$E_1 \cup E_2=\{1,2,\ldots,s\}$. When $C_1 \neq C_2$, the exact number of conjugate pairs is
$\displaystyle{\prod_{i\in E_1 \cap E_2} (2^{n_i}-2)}$. When $C_1 = C_2$, we halve the number.

Suppose that we add the assumption that $p_1(x)=x+1$ and keep $p_i(x)$ primitive for $i\geq 2$. If $1 \in E_1 \cap E_2$, \ie, $a_1=a_1'=1$,
then there is no conjugate pair between the corresponding pairs of cycles since $n_1=1$, making $2^{n_1}-2=0$. Keeping $p_1(x)=x+1$, we prove a result that generalizes~\cite[Prop. 10]{Chang2017}. 

\begin{proposition}\label{prop:nopair}
Consider the LFSR with an arbitrary characteristic polynomial $q(x)\in \F_2[x]$. If $(x+1) \mid q(x)$, 
then $\bv$ and $\widehat{\bv}$ never belong to the same cycle.
\end{proposition}

\begin{proof}
Let $p_1(x):=(x+1)$ and $q(x):= \prod_{i=1}^{s} p_i^{a_i}(x)$ with $a_i \in \N$. If $[\s] \in \Omega(q(x))$ contains a conjugate pair, then $q(x)$ must be its minimal polynomial. Otherwise, let the minimal polynomial be $p(x)$ with $\deg(p(x)) < \deg(q(x))$. Then there is $k \in \Z$ such that $\s + L^{k} \s$ contains $\bS$ as a state and has characteristic polynomial $p(x)$. The sequence containing $\bS$ as a state, however, must have $q(x)$
as its minimal polynomial, which is a contradiction. Hence, $[\s]$ must have the form $\left[\s_1+ L^{i_2}\s_2 +\ldots+ L^{i_s}\s_s \right]$ for $i_2,\ldots,i_{s} \in \Z$ with $p_i^{a_i}(x)$ being the minimal polynomial of $\s_i$ for all $i$. Thus, $\s + L^{k} \s$ has the form
$L^{i'_1} \s'_1+ L^{i'_2} \s'_s + \ldots + L^{i'_s} \s'_s$ with $p_i^{a_i}(x)$ being the characteristic polynomial of $\s'_i$.

In particular, $L^{i'_1} \s'_1=\s_1+L^{k}\s_1$ must be the sum of two sequences having the same minimal polynomial $(x+1)^{a_1}$. Its period is a power of $2$. By~\cite[Lem. 4.1]{KSSK00}, the degree of the minimal polynomial of $L^{i'_1} \s'_1$ is $< a_1$. Hence, the degree of the minimal polynomial of the resulting sequence $\s + L^{k} \s$ must be $< \deg(q(x))$. Thus, it cannot contain $\bS$. \qed
\end{proof}

In joining the cycles, the conjugate pairs
between any cycle and itself are never used. To take advantage of Proposition~\ref{prop:nopair}, let $p_1(x)=x+1$. This implies $\s^1=\1$.
Theorem~\ref{thm:cp} can still be used to determine the conjugate pairs between any two distinct cycles. Let $\s_0$ be as in (\ref{eq:cons}), $\s_1$ and $\s_2$ in (\ref{equ:cc}),
and $\s^1_{d_1}=\s^1_{j_1}=\s^1_{j_1'}=\1$. If $[\s]$ is in $\Omega(f(x))$, then clearly so is $[\s+\1]$.

\begin{corollary}\label{cor:cp}
Let $C_1=[\s_1]$ and $C_2=[\s_2]$ be two cycles in $\Omega(f(x))$ with $p_1(x)=x+1$.
\begin{enumerate}
\item For $C_1$ and $C_2$ to be adjacent, $a_1+a_1'=1$.
\item If $(\bv,\widehat{\bv})$ is a conjugate pair between $C_1=[\s_1]$ and $C_2=[\s_2]$,
then $(\bv+\1,\widehat{\bv}+\1)$ is a conjugate pair between $[\s_1+\1]$ and $[\s_2+\1]$.
\end{enumerate}
\end{corollary}
These facts simplify the determination of the conjugate pairs. This was exhibited in~\cite{Li16} for $f(x)=(x+1) \prod_{j=2}^{s} p_j(x)$ where the $p_j (x)$s are primitive polynomials with pairwise coprime periods.
Corollary \ref{cor:cp} tells us that the same applies to a larger class of polynomials. 

\section{Two Preparatory Algorithms}\label{sec:PrepAlg}

This section discusses two auxiliary algorithms to prepare for the main algorithm. Any conjugate pair can be written as $(\bv,\bv+\bS)$. We have already constructed $\P$ and found a state of each cycle in $\Omega(p_i(x))$ for all $1 \leq i \leq s$ by Algorithm \ref{algo:state}. We now use Algorithm \ref{algo:S} to find the representation
\begin{equation}\label{rep:S}
\bS=\left(T^{c_1}\a^1_{d_1},T^{c_2}\a^2_{d_2},\ldots,T^{c_s}\a^s_{d_s}\right) \P.
\end{equation}

\begin{algorithm}[h!]
\caption{Representing  the Special State $\bS$}
\label{algo:S}
\begin{algorithmic}[1]
\renewcommand{\algorithmicrequire}{\textbf{Input:}}
\renewcommand{\algorithmicensure}{\textbf{Output:}}
	\Require $\P$, $\bS=(1,0,\ldots,0)\in\F_2^n$.
	\Ensure $((c_1,d_1),\ldots,(c_s,d_s))$ such that (\ref{rep:S}) holds.
	\State $(\bv_1,\ldots,\bv_s)\gets \bS \P^{-1}$
	\For {$i$ from $1$ to $s$}
	\For {$j$ from 0 to $t_i-1$}
	\For {$k$ from 0 to $e_i-1$}	
	\If {$\bv_i=\a^i_j$}
	\State {store $(c_i,d_i)\gets(k,j)$}
	\Else
	\State {$\a^i_j \gets T \a^i_j$}
	\EndIf
	\EndFor
	\EndFor
	\EndFor
	\State{return $((c_1,d_1),\ldots,(c_s,d_s))$}
\end{algorithmic}
\end{algorithm}

The algorithm is straightforward, using $\bS \P^{-1} =(\bv_1,\ldots,\bv_s)$
to explicitly find a state $\bv_i$ for each $i$ belonging to a sequence in $\Omega(p_i(x))$. The tuple $(c_i,d_i)$ satisfying $\bv_i=T^{c_i}\a^i_{d_i}$ is found after at most $\sum_{i=1}^{s}(2^{n_i}-1)$ searches. A comparable algorithm accomplishing the same task was presented as \cite[Alg.~1]{Li16}. The latter needs at most $\prod_{i=1}^{s}(2^{n_i}-1)$ searches to find the required representation of $\bS$. Algorithm~\ref{algo:S} is clearly more efficient.

\begin{algorithm}[h!]
\caption{Determining ``Conjugate Pairs" between $2$ Nonzero Cycles in $\Omega(p_i(x))$}
\label{algo:cp-irr}
\begin{algorithmic}[1]
\renewcommand{\algorithmicrequire}{\textbf{Input:}}
\renewcommand{\algorithmicensure}{\textbf{Output:}}
	\Require $\a^i_0,\ldots,\a^i_{t_i-1}$, $T^{c_i}\a^i_{d_i}$, $e_i$.
	\Ensure $\{(\ell_i,-m_i)\}$ satisfying $T^{\ell_i} \a_{j}^{i}+T^{-m_i} \a_{k}^{i}=T^{c_i}\a^i_{d_i}$.
	\State {$K_i \gets \varnothing$} \Comment{initiate the list of defining pairs}
	\For {$j$ from $0$ to $t_i-2$}
	\For {$k$ from $j+1$ to $t_i-1$}
	\For {$(\ell_i, m_i) \in \{0,1,\ldots,e_i-1\} \times \{0,1,\ldots,e_i-1\}$} 
	\If {$T^{\ell_i} \a_{j}^{i}+T^{-m_i} \a_{k}^{i}=T^{c_i}\a^i_{d_i}$}
	\State {append $(\ell_i,-m_i)$ to $K_i$}
	\EndIf
	\EndFor
	\EndFor
	\EndFor
	\State{return $K_i$ and $\lambda_i \triangleq |K_i|$}
\end{algorithmic}
\end{algorithm}

Next comes a crucial step of dividing the problem of determining ``global'' conjugate pairs between any two cycles in $\Omega(f(x))$ by first listing all possible ``local'' candidates between any two cycles in $\Omega(p_i(x))$ for all $i$. Consider $\Omega(p_i(x))$ for a fixed $i$ and choose any two cycles in it. Algorithm \ref{algo:cp-irr} finds all possible pairs $(\bw_1^i,\bw_2^i)$ satisfying $\bw_1^i+\bw_2^i=T^{c_i}\a^i_{d_i}$. Given the representation of $\bS$, running this algorithm for all $\Omega(p_i(x))$ yields the required pairs of states in $\Omega(p_i(x))$. The sum of the two states is now the corresponding state in $\bS \P^{-1}$. The input consists of $t_i$ states, each corresponding to one of the $t_i$ nonzero cycles in $\Omega(p_i(x))$. The output is the set $K_i$ of all tuples $(\ell_i,-m_i)$ that ensure $T^{\ell_i} \a_{j}^{i}+T^{-m_i} \a_{k}^{i}=T^{c_i}\a^i_{d_i}$, \ie, the corresponding state pairs sum to $T^{c_i}\a^i_{d_i}$. The specific choice of a state belonging to a cycle affects neither the number of conjugate pairs nor the states being paired in each conjugate pair. 

For a chosen pair $C_1$ and $C_2$, if Algorithm \ref{algo:cp-irr} yields a defining pair $(x_1,x_2)$, then the defining pair for a ``conjugate pair" between $C_2$ and $C_1$ will be $(x_2,x_1)$. Hence, for each $p_i(x)$, it suffices to take $\binom{t_i}{2}$, instead of $t_i^{2}$, distinct $(j,k)$ tuples. The total number $\lambda_i$ of suitable $(\ell_i,-m_i)$ tuples is a cyclotomic number with specific parameters. 

Running the algorithm is unnecessary for $C_1=[\0]$ and $C_2$ any nonzero cycle. The pair of cycles $(\0,\a^i_j)$ shares a unique conjugate pair if and only if $\a^i_j=\a^i_{d_i}$, \ie, $K_i= \{(0,c_i)\}$. Similarly, if $p_i(x)$ is primitive, one does not need the algorithm. The desired results can be computed using the Zech logarithm $\tau_{n}(\ell)$.
If $\a$ is an $n$-stage state of an $m$-sequence, then the shift-and-add property in~\cite[Thm. 5.3]{GG05} says that $\a+T^{\ell}\a=T^{\tau_n(\ell)}\a$ when $\ell \neq 0$.
If, in $(\bv_1,\ldots,\bv_s)=\bS \P^{-1}$, $\bv_i=T^c\a$, then an output $(y,y')$
of Algorithm \ref{algo:cp-irr} for $p_i(x)$ satisfies $T^y\a+T^{y'}\a=T^c\a$. Hence,
$T^{y'}\a=T^c\a+T^y\a=T^c(\a+T^{y-c}\a)=T^{c+\tau_n(y-c)}\a $. Thus, the desired output must be $\{(y,c+\tau_n(y-c))~|~ y \in \{0,1,\ldots,2^n-2\} \setminus \{c\} \}$
and knowing $\tau_{n}(y)$ is sufficient to find the tuples.

\section{The Main Algorithm}\label{sec:MainAlg}

The ingredients to construct the adjacency graph $G$ associated with $f(x)$ is now ready. Algorithm \ref{algo:all} uses the results of Algorithm \ref{algo:cp-irr} to determine all conjugate pairs between any pair of distinct nonzero cycles. Let us assume that $C_1=[\s_1] \neq C_2=[\s_2]$ in $\Omega(f(x))$ are given in the form specified by (\ref{equ:cc}) with the states $\bv_1$ of $C_1$ and $\bv_2$ of $C_2$ written as
\begin{equation}\label{state:v12}
\bv_1 =\left(a_1T^{\ell_1}\a_{j_1}^{1}, \ldots, a_sT^{\ell_{s}}\a_{j_s}^{s}\right)\P \mbox{ and }
\bv_2 =\left(a'_1T^{\ell_1'}\a_{j_1'}^{1}, \ldots, a_s'T^{\ell_{s}'}\a_{j_s'}^{s}\right)\P.
\end{equation}

\begin{algorithm}[h!]
\caption{All Conjugate Pairs between $2$ Nonzero Cycles in $\Omega(f(x))$}
\label{algo:all}
\begin{algorithmic}[1]
\renewcommand{\algorithmicrequire}{\textbf{Input:}}
\renewcommand{\algorithmicensure}{\textbf{Output:}}
\Require  $\bv_1, \bv_2$ defined in (\ref{state:v12}) and $K_1,K_2,\ldots,K_s$ from Algorithm~\ref{algo:cp-irr}.
\Ensure All conjugate pairs between distinct nonzero cycles $C_1$ and $C_2$.
\For{$i$ from $1$ to $s$}
	\If{$\lambda_i=0$}
	\State{return: there is no conjugate pair; break}
	\EndIf
\EndFor
\State{$CP \gets \varnothing$} \Comment{initiate the set of conjugate pairs}
\For{each $((u_{k_1},v_{k_1}),\ldots,(u_{k_s},v_{k_s}))$ in $K\triangleq K_1 \times K_2 \times \ldots \times K_s$} \Comment{$1 \leq k_i \leq \lambda_i$ for $1 \leq i \leq s$}
\If{
	$\begin{cases}
	u_{k_i}-\ell_i \equiv u_{k_m}-\ell_m \\
	v_{k_i}-\ell_i' \equiv v_{k_m}-\ell_m'
	\end{cases}$ holds in modulo $\gcd(f_i,f_m)$ for all $1 \leq m < i \leq s$}
	\State{$\bv \gets (a_1T^{u_{k_1}}\a_{j_1}^{1}, \ldots, a_sT^{u_{k_s}}\a_{j_s}^{s})\P$}
	\State{append the conjugate pair $(\bv,\widehat{\bv})$ to $CP$}
\EndIf
\EndFor
\State{return $CP$}
\end{algorithmic}
\end{algorithm}

\begin{theorem}\label{th:alg}
Algorithm~\ref{algo:all} is correct.
\end{theorem}

\begin{proof}
If $C_1$ and $C_2$ are adjacent, then there are integers $\ell$ and $\ell'$ satisfying $T^{\ell}\bv_1+T^{\ell'}\bv_2=\bS$. In particular, for each $i$, it holds that $a_iT^{\ell+\ell_i}\a_{j_i}^{i}+a'_iT^{\ell'+\ell_i'}\a_{j_i'}^{i}=T^{c_i}\a^i_{d_i}$. Hence, $(\ell+\ell_i,\ell'+\ell_i')$ must be a pair $(u_{k_i},v_{k_i})$ obtained from Algorithm \ref{algo:cp-irr}, \ie, $\ell+\ell_i\equiv u_{k_i}\Mod{f_i}\mbox{ and } \ell'+\ell_i'\equiv v_{k_i}\Mod{f_i}$, implying
\[
\begin{cases}
\ell \equiv u_{k_i}-\ell_i \Mod{e_i}  \mbox{ for all } i \in E_1\\
\ell'\equiv v_{k_i}-\ell_i' \Mod{e_i} \mbox{ for all } i \in E_2
\end{cases}.
\]
We know that $\ell$ and $\ell'$ exist if and only if the systems of congruences
\[
\begin{cases}
\ell  \equiv u_{k_i}-\ell_i \Mod{f_i} \mbox{ for all } 1 \leq i \leq s\\
\ell' \equiv v_{k_i}-\ell_i'\Mod{f_i} \mbox{ for all } 1 \leq i \leq s
\end{cases}
\]
can be simultaneously solved. For this to happen, Theorem~\ref{CRT} requires that the systems
\[
\begin{cases}
u_{k_i}-\ell_i  \equiv u_{k_j}-\ell_j\Mod{\gcd(f_i,f_j)}\\
v_{k_i}-\ell_i' \equiv v_{k_j}-\ell_j'\Mod{\gcd(f_i,f_j)}
\end{cases}
\]
hold for $1 \leq i \neq j \leq s$.
	
Algorithm \ref{algo:all} checks whether this requirement is met. Given $2 \leq i \leq s$, the verification is performed for all $1 \leq m < i$. All relevant congruences are certified to hold simultaneously. The process is terminated at the first instance when one of the congruences fails to hold. For a chosen set $\{(u_{k_i},v_{k_i})\}_{i=1}^{s}$, once the systems of congruences are certified to hold, then $\ell$ and $\ell'$ exist, making $(T^{\ell}\bv_1,T^{\ell'}\bv_2)$ a conjugate pair. 
Since $\ell+\ell_i\equiv u_{k_i}\Mod{f_i}$, one confirms that $(\bv,\widehat{\bv})$ is a conjugate pair with $\bv=(a_1T^{u_{k_1}}\a_{j_1}^{1}, a_2T^{u_{k_2}}\a_{j_2}^{2}, \ldots, a_sT^{u_{k_s}}\a_{j_s}^{s})\P$. \qed
\end{proof}

In~\cite{Li16}, the $p_i(x)$s are primitive and $e_1,\ldots,e_s$ are pairwise coprime. The conjugate pairs can be found without Algorithm \ref{algo:all}.
Let $C_1=[a_1 \s_1+\ldots+a_s \s_s]$ and $C_2=[b_1 \s_1+ \ldots+ b_s \s_s]$ with $a_i \vee b_i=1$.
Given $\bS \P^{-1}=(\a_1,\ldots,\a_s)$, it was shown that $\bv=(\bv_1,\ldots,\bv_s)\P$ in the conjugate pair
$(\bv,\widehat{\bv})$ has
\[
\bv_i=
\begin{cases}
\0 &\mbox{ if } a_i=0 \mbox{ and } b_i=1,\\
\a_i &\mbox{ if } a_i=1 \mbox{ and } b_i=0,\\
\bv_i \in \F_2^{n_i} \setminus\{\0,\a_i\} &\mbox{ if } a_i=b_i=1.
\end{cases}
\]
If the periods are not coprime, then the situation is more involved.
One must ensure that the required systems of congruences hold simultaneously.
Algorithm~\ref{algo:all} covers this more general situation. The overall running time can be further improved by using the properties in Section~\ref{sec:graph} to rule out pairs of cycles with no conjugate pairs prior to running Algorithm~\ref{algo:all}.

Since Algorithm~\ref{algo:all} builds upon the results of Algorithm \ref{algo:cp-irr}, sufficient storage must be allocated to have them ready at hand. The two algorithms can be merged. In Algorithm~\ref{algo:all}, at the precise step when the ``conjugate pairs'' between the specified two cycles in $\Omega(p_i(x))$ are needed, a search for them can be executed to cut down on the storage requirement. Anticipating the need to find all conjugate pairs between any two cycles in $\Omega(f(x))$ in various applications and noting that the data are reasonably small, we prefer storing them.

Algorithm~\ref{algo:all} transforms the problem of finding conjugate pairs between two cycles with a more complicated characteristic polynomial into one of finding ``conjugate pairs'' between two cycles with simpler characteristic polynomials. Once we have determined the ``conjugate pairs'' for each $p_i(x)$, solving systems of congruences leads us to the desired conjugate pairs for $f(x)$. In contrast, for two cycles with minimal polynomial $f(x)$, exhaustive search requires around $(\lcm(e_1,\ldots,e_s))^2$ computations. 

Let $\psi$ denote the number of cycles in $\Omega(f(x))$. The number of pairs of nonzero cycles in Algorithm~\ref{algo:all} is $\displaystyle{\binom{\psi-1}{2}}$. For each of them, the expected cardinality of $K$ is $\displaystyle{\prod_{i=1}^s \lambda_i = \prod_{i=1}^s \left \lceil \frac{e_i}{t_i} \right \rceil }$. For any element of $K$, performing the CRT takes $\mathcal{O}(\log^2(F))$ where $F:=\prod_{i=1}^s f_i$. Here we use the analysis on the complexity of CRT based on Gardner's Algorithm \cite[Sect.~4.3.2]{Knuth2}. A similar analysis can also be found in~\cite[Ch.~3]{Ding96}.

\section{The Case of Repeated Roots}\label{sec:reproots}

This section briefly considers characteristic polynomials with repeated roots
\[
q(x):=\prod_{i=1}^s p_i^{b_i}(x) \mbox{ with } b_i > 1 \mbox{ for some } i.
\]
We have recently determined the cycle structure of $\Omega(q(x))$ in~\cite[Thm.~ 1]{CELW16}. Hence, only a brief outline of the idea is provided here. The number of cycles in $\Omega(p_i^{b_i}(x))$ can be derived from~\cite[Thm. 8.63]{LN97}.
\begin{enumerate}
\item The only cycle with period $1$ is $[\0]$.

\item There are $t_{i}$ cycles containing sequences with period $e_i$.

\item Let $\chi_{i}$ be the smallest positive integer such that $2^{\chi_{i}} \geq b_i$ and let $ 1 \leq j < \chi_{i}$. Sequences with period $e_i \cdot 2^{j}$ are partitioned into $\rho$ cycles whenever $b_i > 2$ while those with period $e_i \cdot 2^{\chi_{i}}$ whenever $b_i \geq 2$ are partitioned into $\zeta$ cycles where
\[
\rho:= \frac{2^{n_i \cdot 2^{j}}-2^{n_i \cdot 2^{j-1}}}{e_i \cdot 2^{j}} \mbox{ and }
\zeta:=\frac{2^{n_i \cdot b_i}-2^{n_i \cdot 2^{\chi_{i}-1}}}{e_i \cdot 2^{\chi_{i}}}.
\]
\end{enumerate}

Based on the states of the cycles in $\Omega(p_i(x))$ we give a detailed procedure to derive 
the states of the remaining cycles in $\Omega(p_i^{b_i}(x))$. For brevity, all states are 
considered to have the same length $b_i \cdot n_i$. Using the cycles in $\Omega(p_i^{b_i}(x))$ 
for all $i$, we can determine all cycles in $\Omega(q(x))$ and prove results similar to the 
ones in Lemma~\ref{lem:cycle-f}.
Once the cycle structure of $\Omega(q(x))$ is known, one can use the methods discussed above to study the properties of the conjugate pairs.

Copying the construction of $\P$, we build $\widetilde{\P}$ by replacing the original polynomial $p_i(x)$ by $p_i^{b_i}(x)$. We know from \cite[Lemma 7]{CELW16} that $\widetilde{\P}$ has full rank and can use it to determine
the new representation of $\bS$ as $(\bv_1,\ldots,\bv_s) \widetilde{\P}$. Here, $\bv_i \in \F_2^{b_i \cdot n_i}$ is a state
of a cycle in $\Omega(p_i^{b_i}(x))$ with minimal polynomial $p_i^{b_i}(x)$. Given any two cycles, we deploy Algorithm \ref{algo:cp-irr} with inputs the states of all cycles in $\Omega(p_i^{b_i}(x))$ to output pairs of states $(\x_i,\y_i)$ satisfying $\x_i+\y_i=\bv_i$.
Now, this is where complication arises
since we are no longer able to leverage on tools or results from relevant cyclotomic numbers, plus there are likely to be a lot of
distinct cycles in $\Omega(p_i^{b_i}(x))$.

The proof of the following result on the pair $(\x_i,\y_i)$ is straightforward. Using the cycle structure of $\Omega(p_i^{b_i}(x))$ and some properties of LFSRs, it may be possible to derive more results on the pair $(\x_i,\y_i)$. We leave them for future investigation.

\begin{proposition}\label{rep-root1}
Let $\x_i$ and $\y_i$ be states belonging to respective cycles $C_1=[\s_1]$ and $C_2=[\s_2]$ in $\Omega(p_i^{b_i}(x))$. If there exist $a,b \in \Z$ satisfying $T^a {\x}_i + T^b{\y}_i = \bv_i$, then at least one of $\s_1$ or $\s_2$ must have $p_i^{b_i}(x)$ as minimal polynomial. If $p_1(x)= 1+x $, then exactly one of $\s_1$ and $\s_2$ has minimal polynomial $(1+x)^{b_i}$.
\end{proposition}

Algorithm~\ref{algo:all} can also be used to find all conjugate pairs between any two cycles in $\Omega(q(x))$. One needs to exercise greater care
here since the cycles in $\Omega(p_i^{b_i}(x))$ may have distinct periods. This affects
the application of the generalized CRT (Theorem \ref{CRT}).

Algorithms \ref{algo:cp-irr} and \ref{algo:all} perform better for large $s$, \ie,
when $f(x)$ has more distinct irreducible polynomials as its factors.
Algorithm \ref{algo:cp-irr} is more efficient when there are less cycles and the periods are small.
Since $q(x)$ has repeated roots, one has to treat the cycles in $\Omega(p_i^{b_i}(x))$ separately according to their respective periods.
A modified version of Algorithm \ref{algo:cp-irr} takes more time here since the overall degree is $b_i \cdot n_i$ and there are a lot of cycles to pair up. This typically drives the complexity of finding pairs of states $(\x_i,\y_i)$ satisfying $\x_i+\y_i=\bv_i$ much higher than in the case $\Omega\left(\prod_{i=1}^t p_i(x)\right)$ where $p_i(x)$ for $1 \leq i \leq t$ are distinct irreducible polynomails with the same overall degree $b_i \cdot n_i$ where the generalized CRT simplifies the process significantly. 
	
In short, we can slightly modify our approach for $f(x)$ to work on $q(x)$. A characteristic polynomial $q(x)$ with repeated roots can be used to generate de Bruijn sequences using the already mentioned modification on the respective algorithms. Unless one is prepared to commit much more computational resources or is required to produce more de Bruijn sequences that can be constructed based on $f(x)$, using $q(x)$ is generally not advisable.

In the literature, studies on the case of characteristic polynomials with repeated roots have been quite limited,
{\it e.g.}, $q(x)=(x+1)^n$ in \cite{Hemmati84} and $q(x)=(x+1)^a p(x)$ with $p(x)$ having no repeated roots or is primitive done, respectively, in~\cite{Li14-1,Li2017}. 
Prior studies looked into cases with $(x+1)^b \mid q(x)$ for $b \in \N$ because their cycle structures and adjacency graphs had been well-established.

\section{Generating the de Bruijn Sequences}\label{sec:generate}

After implementing Algorithms~\ref{algo:state} to \ref{algo:all} we obtain all of the conjugate pairs between
any two adjacent cycles in $\Omega(f(x))$ and, hence, the adjacency graph $G$.
The edges between a specific pair of vertices in $G$ correspond to the conjugate pairs shared by the represented cycles.
Derive the graph $\widehat{G}$ from $G$ by bundling together multiple edges incident to the same pair of vertices into one edge. An edge
in $\widehat{G}$ corresponds to a set of conjugate pair(s). Next, we use Algorithm~\ref{algo:sptree} to generate all spanning trees in
$\widehat{G}$. Line~\ref{line:span} requires that each identified tree contains all vertices in $\widehat{G}$. The {\tt for loop} in Lines~\ref{line:start} to \ref{line:end} ensures that all vertices are visited and eventually checked and that no cycle occurs. Note that the set $B$ and, hence, $X$ may be empty. Depending on the specifics of the input graph and the choice of $\overline{V}$, given the current sets $VList$ and $VCheck$, there are $2^{|B|} <  2^{\deg(\overline{V})}$ choices for $X$.

\begin{algorithm}[ht!]
\caption{Finding all Spanning Trees in $\widehat{G}$}
\label{algo:sptree}
\begin{algorithmic}[1]
\renewcommand{\algorithmicrequire}{\textbf{Input:}}
\renewcommand{\algorithmicensure}{\textbf{Output:}}
\Require $\widehat{G}$ with $V(\widehat{G}):=\{V_1,\ldots,V_{\psi}\}$
and edge set $E(\widehat{G})$.
\Ensure All spanning trees in $\widehat{G}$.
\State $VList \gets \{V_1\}$ \Comment{initiate the list of visited vertices; for convenience, $V_1:=[\0]$}
\State $VCheck \gets \varnothing$ \Comment{initiate the list of checked vertices}
\State $EList \gets \varnothing$ \Comment{initiate the list of collected edges}
\Procedure{SpanTrees~}{$ST(VList,VCheck,EList)$}
	\If{$|EList|=\psi-1$} \label{line:span}
		\State{return $EList$}
	\Else
		\State{take any vertex $\overline{V} \in (VList \setminus VCheck)$}
		\State{append $\overline{V}$ to $VCheck$}	
		\State{$B \gets \{V_{j} \mbox{ satisfying } (\overline{V},V_{j}) \in E(\widehat{G}) \mbox{ with } V_{j} \notin VList\}$}
		\For{each $X \subseteq B$} \label{line:start}
			\State{$VList \gets VList \cup X$}
			\State{$EList \gets EList \cup \{(\overline{V},V_{k}) \mbox{ with } V_{k} \in X\}$}
			\State{$ST(VList,VCheck,EList)$}
		\EndFor \label{line:end}
	\EndIf
\EndProcedure
\end{algorithmic}
\end{algorithm}

Our implementation is in \texttt{python}. Readers who prefer to code in \texttt{C} may opt to use
D.~Knuth's implementation of~\cite[Alg. S, pp.~464 ff.]{Knuth4A} named \texttt{grayspspan}~\cite{Gray}. Its running time is roughly estimated to be $\O(\mu+\psi+\zeta_{\widehat{G}})$ where
$\mu$ is the total number of edges in $\widehat{G}$, $\psi$ is the number of cycles in $\Omega(f(x))$,
and $\zeta_{\widehat{G}}$ is the number of the spanning trees in $\widehat{G}$. Notice that $\mu < 2^{n-1}$,
$\psi \leq \displaystyle{\frac{1}{n}\sum_{d \mid n}\phi(d)2^{\frac{n}{d}}}$ (see, \eg, \cite{M72}), and
$\zeta_{\widehat{G}}$ is much larger than both $\mu$ and $\psi$. Hence, $\O(\mu+\psi+\zeta_{\widehat{G}})=\O(\zeta_{\widehat{G}})$. 
We ran numerous simulations and came to the conclusion that the running time of Algorithm~\ref{algo:sptree} is the same as that of Algorithm S. We prefer the procedure in Algorithm~\ref{algo:sptree} since it leads to a simpler mechanism, both deterministic and random, on how to pick a particular spanning tree to use in the actual generation of the sequences. 
Next comes the cycle joining procedure.

Let $\Upsilon_{G}$ be a chosen spanning tree in $G$ and $\{(\bv_1,\widehat{\bv_1}),\ldots,(\bv_{\psi-1},\widehat{\bv}_{\psi-1})\}$ be its edge set. Let $\overline{\bv}$ denote the last $n-1$ bits of any $\bv \in \F_2^n$. Define a new set $E(\Upsilon_G):= \{\overline{\bv}_1,\ldots,\overline{\bv}_{\psi-1}\}$. From an arbitrary initial state (any vector in $\F_{2}^{n}$), use the LFSR with feedback function $h(x_0,\ldots,x_{n-1})$, \ie, with characteristic polynomial $f(x)$, 
to generate $a_0,a_1,\ldots,a_{n-1},a_n,\ldots$. For each input $(a_i,\ldots,a_{i+n-1})_{i \geq 0}$, let
\[
\begin{cases}
a_{i+n}=h(a_{i},\ldots,a_{i+n-1}) & \mbox{if } (a_{i+1},\ldots,a_{i+n-1}) \notin E(\Upsilon_G),\\
a_{i+n}=h(a_{i},\ldots,a_{i+n-1})+1 & \mbox{if } (a_{i+1},\ldots,a_{i+n-1}) \in E(\Upsilon_G).
\end{cases}
\]
The resulting sequence is de Bruijn with feedback function 
\[	h(x_0,\ldots,x_{n-1})+\sum_{\mathbf{w} \in E(\Upsilon_G)}\prod_{i=1}^{n-1}(x_i + w_i+1).
\] 
Performing the cycle joining procedure to all spanning trees in $G$ gives us all de Bruijn sequences in this class.

To tie up all examples pertaining to $f(x)=(x+1)(x^2+x+1)(x^4+x^3+x^2+x+1)$, label the vertices in $G$ as $V_1,V_2,\ldots,V_{16}$ according to the ordering in Table~\ref{table:ex}. Since $\bS\P^{-1}=(1,1,1,1,0,1,0)=\left(\1,L^{2}\s^2_0,L^{2}\s^3_2\right)$ is in $\left[\1+\s^2_0+\s^3_2\right]$, $\bS$ is the initial state of $L^2 \left(\1+\s^2_0+\s^3_2\right)$. Using the appropriate results from Example~\ref{example2} and running Algorithm~\ref{algo:all} yield all the conjugate pairs between any two cycles. We summarize the count in Table~\ref{table:edge}. The complete adjacency graph $G$ has edges labeled by the conjugate pairs.

\begin{table}[h!]
\caption{Number of Conjugate Pairs}
\label{table:edge}
\renewcommand{\arraystretch}{1.2}
\centering
\begin{tabular}{cc|cc|cc|cc|cc}
\hline
Edge & $\#$ & Edge & $\#$ & Edge & $\#$ & Edge & $\#$ & Edge & $\#$ \\
\hline
		
$\{V_1,V_{16}\}$ & $1$ & $\{V_2,V_{13}\}$ & $1$ &
$\{V_2,V_{16}\}$ & $2$ & $\{V_3,V_{14}\}$ & $2$ & $\{V_3,V_{15}\}$ & $1$ \\
		
$\{V_3,V_{16}\}$ & $2$ & $\{V_4,V_{14}\}$ & $1$ & 
$\{V_4,V_{15}\}$ & $2$ & $\{V_4,V_{16}\}$ & $2$ & $\{V_5,V_{10}\}$ & $1$ \\
		
$\{V_5,V_{14}\}$ & $2$ & $\{V_5,V_{15}\}$ & $2$ &
$\{V_6,V_{11}\}$ & $2$ & $\{V_6,V_{12}\}$ & $1$ & $\{V_6,V_{13}\}$ & $2$ \\
		
$\{V_6,V_{14}\}$ & $4$ & $\{V_6,V_{15}\}$ & $2$ & 
$\{V_6,V_{16}\}$ & $4$ & $\{V_7,V_{11}\}$ & $1$ & $\{V_7,V_{12}\}$ & $2$ \\
		
$\{V_7,V_{13}\}$ & $2$ & $\{V_7,V_{14}\}$ & $2$ & 
$\{V_7,V_{15}\}$ & $4$ & $\{V_7,V_{16}\}$ & $4$ & $\{V_8,V_{9}\}$ & $1$ \\
		
$\{V_8,V_{10}\}$ & $2$ & $\{V_8,V_{11}\}$ & $2$ &
$\{V_8,V_{12}\}$ & $2$ & $\{V_8,V_{14}\}$ & $4$ & $\{V_8,V_{15}\}$ & $4$ \\
		
\hline
\end{tabular}
\end{table}

Let $\M_{1}$ be the diagonal $8 \times 8$ matrix with entries, in order, $1,3,5,5,5,15,15,15$. Then 
\[
\M =
\begin{pmatrix}
\M_{1} & \M_{2} \\
\M_{2} & \M_{1}
\end{pmatrix} \mbox{ with }
\M_{2} =
\begin{pmatrix}
0 & 0 & 0 & 0 & 0 & 0 & 0 & -1 \\
0 & 0 & 0 & 0 & -1 & 0 & 0 & -2 \\
0 & 0 & 0 & 0 & 0 & -2 & -1 & -2 \\
0 & 0 & 0 & 0 & 0 & -1 & -2 & -2 \\
0 & -1 & 0 & 0 & 0 & -2 & -2 & 0 \\
0 & 0 & -2 & -1 & -2 & -4 & -2 & -4 \\
0 & 0 & -1 & -2 & -2 & -2 & -4 & -4 \\
-1 & -2 & -2 & -2 & 0 & -4 & -4 & 0
\end{pmatrix}.
\]
The cofactor of any entry in $\M$ gives $12,485,394,432 \approx 2^{33.54}$ as the number of constructed de Bruijn sequences of order $7$. The matrix representation $\widehat{\M}$ of $\widehat{G}$ has main diagonal entries $1,2,3,3,3,6,6,6,1,2,3,3,3,6,6,6$. For $i \neq j$, the $(i,j)$ entries is $-1$  whenever $V_i$ and $V_j$ are adjacent and $0$ otherwise. Hence, $\zeta_{\widehat{G}}=1,451,520 \approx 2^{20.47}$, which is easier to process computationally than the number $\zeta_{G}\approx 2^{33.54}$ of spanning trees in $G$.

The polynomial $f(x)$ studied in~\cite{Li16} is the product of $x+1$ and some primitive polynomials with coprime periods. Such a choice greatly simplifies the procedure since the resulting cycle structure is simple and the conjugate pairs
can be deduced directly. Their main motivation was to be able to generate the sequences quickly. The main drawback is the relatively low number of sequences generated. The construction in \cite[Sect. IV]{Li16}, for example, only produces sequences from a special spanning tree named {\it maximal spanning tree}. The number of de Bruijn sequences generated by \cite[Alg.~1]{Li16} is $\displaystyle{\O(2^{(2^{s-1}-1)n})}$. From $f(x)=(x+1)p_1(x)p_2(x)$, the number of sequences generated by \cite[Alg.~2]{Li16} is $\O(2^{3n})$. To illustrate the point more concretely, the number of all de Bruijn sequences of order $8$ that can be generated from $f(x)=(x+1)(x^3+x^2+1)(x^4+x^3+1)$ is $926,016$. Using the same $f(x)$,~\cite[Alg.~2]{Li16} produces only $592,704$ of them.

\section{Implementation}\label{sec:implem}
\begin{table}[h!]
\caption{Some Implementation Results}
\label{table:implem}
\renewcommand{\arraystretch}{1.3}
\setlength{\tabcolsep}{0.18cm}
\centering
\begin{tabular}{c c l l ccc r}
\hline
No. & $n$ &  $\{p_i(x)\}$  & $\{e_i\}$ & $\psi$ & $\zeta_{G}$ & $\zeta_{\widehat{G}}$ & $Run$ \\
\hline

$1$ & $6$ & $\{1011,1101\}$  & $\{7,7\}$ & $10$ & $393~216$ & $51~984$ & $3.57$s \\

$2$ & $7$ & $\{11,111,11111\}$ &  $\{1,3,5\}$ & $16$ & $2^{33.5}$ & $2^{20.5}$ & $1.89$s \\ 

$3$ & $8$  & $\{11,1101,11001\}$ &  $\{1,7,15\}$ & $8$ & $926~016$ & $15$ & $2.77$s \\ 

$4$ & $8$ &  $\{10011, 11111\}$  & $\{15,5\}$ & $20$ & $2^{60.8}$ & $ 2^{53.0}$ & $16.37$s \\

$5$ & $9$ & $\{111, 1011, 11111\}$ & $\{3,7,5\}$ & $16$  & $ 2^{54.4}$  & $ 2^{28.8}$ & $5.47$s \\

$6$ & 	$9$ & $\{11, 100111001\}$ & $\{1,17\}$ & $32$  & $ 2^{113.4}$  & $ 2^{86.7}$ & $7.20$s \\

$7$ & 	$10$ & $\{11,111,1011,11111\}$ & $\{1,3,7,5\}$ & $32$  & $ 2^{116.0}$  & $ 2^{61.1}$ & $16.47$s \\

$8$ & $10$ & $\{11111 11111 1\}$ & $\{11\}$ & $94$  & $ 2^{304.9}$  & $ 2^{299.1}$ & $59.30$s \\

$9$ & $11$ & $\{111, 1011, 1001001\}$ & $\{3,7,9\}$ & $60$  & $ 2^{251.9}$ &  $ 2^{190.0}$ & $2$m $32$s \\

$10$ & $11$ & $\{101011100011\}$ & $\{23\}$ & $90$  & $ 2^{388.8}$  & $ 2^{373.8}$ & $1$m $27$s \\ 

$11$ & 	$12$ & $\{1001001,1010111\}$ & $\{9,21\}$ & $74$  & $ 2^{398.7}$  & $ 2^{350.7}$ & $6$m $09$s \\

$12$ & 	$13$ & $\{11, 111, 11111, 1001001\}$ & $\{1,3,5,9\}$ & $240$ & $2^{1114.6}$ & $2^{853.8}$ & $24$m $12$s \\

$13$ & 	$14$ & $\{111, 11111, 100111001\}$ & $\{3,5,9\}$ & $128$ & $2^{800.2}$ & $2^{583.7}$ & $5$m $04$s \\ 

$14$ & 	$15$ & $\{1001001, 1000000011\}$ & $\{9,73\}$ & $64$ & $2^{508.6}$ & $2^{277.3}$ & $5$m $06$s \\

$15$ & 	$16$ & $\{1001001, 10000001111\}$ & $\{9,341\}$ & $32$ & $2^{274.2}$ & $2^{97.0}$ & $7$m $58$s \\ 

$16$ & 	$16$ & $\{11, 111, 1011, 11111, 1001001\}$& $\{1,3,7,5,9\}$ & $480$ & $2^{2925.8}$ & $2^{1966.8}$ & $15$hr $23$m \\

$17$ & 	$17$ & $\{100 111 111, 1 000 000 011\}$ & $\{85,73\}$ & $32$ & $2^{310.1}$ & $2^{111.3}$ & $15$m $01$s \\

$18$ & 	$18$ & $\{111010111, 10001000111\}$ & $\{17,341\}$ & $64$ & $2^{630.6}$ & $2^{261.5}$ & $37$m $48$s \\

$19$ & 	$19$ & $\{1001100101, 10000110101\}$ & $\{73,93\}$ & $96$ & $2^{1076.3}$ & $2^{530.7}$ & $1$hr $41$m \\
 
$20$ & 	$20$ & $\{11111, 1001001, 10000001111\}$ & $\{5,9,341\}$ & $128$ & $2^{1365.0}$ & $2^{564.4}$ & $4$hr $36$m \\ 
\hline
\end{tabular}
\end{table}

A basic software implementation was written in {\tt python} 2.7 and is available online at \url{https://github.com/adamasstokhorst/debruijn} \cite{EF17a}. The platform was a laptop having 11.6 GB available memory with Windows 10 operating system powered by an Intel i7-7500U CPU 2.70GHz. Table~\ref{table:implem} presents some of the implementation results. 

The program determine \underline{exactly} the number $\psi$ of distinct cycles in $\Omega(f(x))$, the adjacency graph $G$ and its corresponding matrix $\M$, the number $\zeta_{G}$ of de Bruijn sequences that can be generated, and the number $\zeta_{\widehat{G}}$ of all spanning trees in $\widehat{G}$. It then outputs $100$ de Bruijn sequences with initial state $\0$. The total running time is denoted by $Run$. For brevity, $2^{k}$ with $k \geq 20$ rounded to one decimal place approximates $\zeta_G$ and $\zeta_{\widehat{G}}$ and we remove the $\approx$ sign from the table. Entry 2 is our example while Entry 3 is an example in~\cite{Li16}. 

To build a library of de Bruijn sequences of order $n$, one can store the graphs $G$ and generate $\widehat{G}$. Once a tree in $\widehat{G}$ has been identified, the corresponding set of tree(s) in $G$ is listed down. For each of these trees, one uses the cycle joining method to generate all de Bruijn sequences in this class. To generate a random de Bruijn sequence, we can chose a random tree $\Upsilon_{\widehat{G}}$ before selecting one among the corresponding trees in $G$ randomly, say $\Upsilon_{G}$. We then use an arbitrary $\bv \in \F_{2}^{n}$ as the initial state in the cycle joining routine. The randomization does not significantly alter the running time.

A user may want to output one de Bruijn sequence chosen uniformly at random from among all of the constructible sequences. One simply applies Broder's elegant algorithm~\cite[Alg. Generate]{Bro89} on $G$ (not on $\widehat{G}$) to get a uniformly random tree $\Upsilon_{G}$. The algorithm's expected running time per generated tree is $\O(\psi \log \psi)$ for most graphs. The worst case value is $\O(\psi^{3})$. One then picks an arbitrary initial state $\bv$ from $\F_2^{n}$ and performs the cycle joining routine with $\Upsilon_G$ and $\bv$ as ingredients. Once $G$ has been determined, the rest of the process is very fast.

For large $n$ or $s$, the required memory and running time can quickly exceed available resources if we insist on building $G$ completely. To generate a few (at least one) de Bruijn sequences of large order, we perform the routine up to generating all the cycles and finding a state in each cycle. We can then proceed to identify a simple connected subgraph containing all cycles in $V(G)$. One way to do this is to index the cycles as $C_1,C_2,\ldots,C_{\psi}$ before collecting all neighbours of $C_1$ in a set $\mathcal{N}_{C_1}$. Only one conjugate pair between $C_1$ and each member of $\mathcal{N}_{C_1}$ needs to be found. We pick the cycle with the smallest index in $\mathcal{N}_{C_1}$ and determine all of its neighbours in the set $V(G) \setminus (\mathcal{N}_{C_1} \cup C_1)$, listing a single conjugate pair for each pair of cycles. The procedure is repeated until a connected graph containing $V(G)$ is found. The spanning trees in the resulting graph generates de Bruijn sequences.

\section{Conclusion and Future Directions}\label{sec:conclu}

We put forward a method to generate a large class of binary
de Bruijn sequences from LFSRs with an arbitrary characteristic polynomial, paying special attention to the case where $f(x)$ is the product of $s$ pairwise distinct irreducible polynomials. The related structures are studied in details.
Our approach covers numerous prior constructions as special cases. Table~\ref{table:compare} lists their relevant parameters and results arranged chronologically by publication dates. In addition to the cycle structure and the adjacency graph, we note if the feedback functions of the maximum-length NLFSRs are explicitly derived. Discussions on the number of de Bruijn sequences that can be constructed and on the algorithmic steps to generate them are also considered in the list.

\begin{table}[h!]
\caption{List of Prior Cycle-Joining Constructions}
\label{table:compare}
\renewcommand{\arraystretch}{1.1}
\setlength{\tabcolsep}{0.09cm}
\centering
\begin{tabular}{clc cccc c}
\hline

No. & The Characteristic & Ref. & Structure & Graph & Feedback & Number of &   Generating \\

   & Polynomial $f(x)$ & & $\Omega(f(x))$ & $G$ & Function & Sequences &   Algorithm(s) \\
\hline

$1$ & $(1+x)^n$ & \cite{Hemmati84} & Yes & No & No & \multicolumn{2}{c}{Not discussed in the paper}\\

$2$ & $f(x)$ is irreducible & \cite{HH96} & Yes & Yes & No & Bounds & Not discussed\\

$3$ & $(1+x)^m~p(x)$, & \cite{Li14-1} & Yes & Yes & Yes & Estimated & Implicit, \\
& $p(x)$ is primitive & & (All $m$) & \multicolumn{3}{c}{(Only applicable for $m=3$)} & without analysis\\
		
$4$ & $(1+x^3)~p(x)$, & \cite{Li14-2} & Yes & Yes & Yes & Estimated & For some, not\\
 & $p(x)$ is primitive && &&&  & all, sequences\\

$5$ & $\prod_{i=1}^{s} p_i(x)$, $p_i(x)$ primitive, & \cite{Li16} & Yes & Yes & Yes & Estimated & Yes, with analysis  \\
& degrees  increasing & & (All) & \multicolumn{4}{c}{(These results apply only for $p_1(x)=1+x$)} \\
& and pairwise coprime & & &&&   & \\

$6$ & Product of 2 irreducibles & \cite{Chang2017} & Yes & Yes & Yes & Estimated & Yes, without  \\
& && &&&& analysis\\

$7$ & $\ell(x) ~p(x)$, $p(x)$ primitive, & \cite{LL17a} & Yes & Yes & Yes & Exact & Implicit, with analysis \\
& $\ell(x)$ has small degree  & & \multicolumn{5}{c}{(If periods of $\ell(x)=1+x+x^3+x^4$ and $p(x)$ are coprime)} \\

$8$ & $\prod_{i=1}^{s} \ell_i(x)$, factors are& \cite{LL17} &  Yes & Yes & \multicolumn{2}{c}{Not discussed} & Implicit, \\
& pairwise coprime & & &&&&  without analysis\\
\hline
\end{tabular}
\end{table}

The followings are useful details on each entry in Table~\ref{table:compare}. 
\begin{enumerate}
\item The number of conjugate pairs between any two cycles in~\cite{Hemmati84} is shown to be $\leq 2$ but without any method to determine them. A complete adjacency graph is not provided.

\item A lower bound on the number of resulting sequences is stated in~\cite[Thm.~6]{HH96}. Exact values in some specific cases are given in~\cite[Thm.~7 and Thm.~8]{HH96}.
\item The estimated number $\mathcal{O}(2^{4n})$ of the resulting sequences in~\cite{Li14-1} applies for $m=3,n \geq 5$.

\item Not all possible de Bruijn sequences are constructed in~\cite{Li14-2} since only special types of spanning trees are used in the cycle joining process. Using this well chosen adjacency subgraph, the time as well as memory complexity of the algorithm that generate $\mathcal{O}(2^{4n})$ and $\mathcal{O}(2^{8n})$ de Bruijn sequences for odd and even $n$, respectively, is $\mathcal{O}(n)$.

\item Recall that $n_s$ is the degree of $p_s(x)$. When $p_1(x)=1+x$, the number of de Bruijn sequences that can be constructed in~\cite{Li16} is estimated to be $\mathcal{O}\left(2^{(2^k-1)n}\right)$. The algorithm runs in memory complexity $\mathcal{O}\left(2^{s+1} (s+1) n\right)$ and time complexity $\mathcal{O}\left(2^{n-n_s} (s+1) n\right)$. If $s=3$ and $n \geq 8$, the time complexity can be reduced to $\mathcal{O}\left(n^{\log \log(n)}\right)$.

\item The exact number of de Bruijn sequences that can be constructed and the discussion on the time complexity $\mathcal{O}(n)$ required to completely generate one such sequence are given in~\cite[Sect. VII]{LL17a}. 

\item Only a rough estimate on the number of constructed sequences is given in~\cite[Eq. (21)]{Chang2017}. An exact count is provided in~\cite[Thm.~4]{Chang2017} for $f(x)=(1+x+x^2) p(x)$ where $p(x)$ is any primitive polynomial  of degree $\geq 3$. 

\item The main focus in~\cite{LL17} is the adjacency graph. The number of sequences that can be built and their feedback functions are out of the scope and the discussion on how to generate them is limited to some illustrative cases.
\end{enumerate}

In this work we build a de Bruijn sequence generator and back our claims with implementation results. Our basic implementation scenarios can be extended to cover more application-dependent purposes. Using the results established in this paper, one can write an implementation software to identify enough conjugate pairs that yield a spanning tree, from which a de Bruijn sequence of a large order can be built quickly.

Many interesting directions remain open for investigation. A partial list includes the following problems.
\begin{enumerate}
\item Optimize the algorithms above to practically generate de Bruijn sequences for larger $n$.

\item Derive a reasonably tight lower bound on the number of de Bruijn sequences of a very large order that can be generated by cycle joining method based on a given LFSR.

\item Generalize our results to any finite field $\F_{q}$. This is nontrivial since the presence of multiple nonzero elements changes the definition of a conjugate pair.

\item Determine good lower and upper bounds on the linear complexity of the resulting sequences or their modified version, \ie, those obtained by deleting a zero from the consecutive string of $n$ zeroes. We believe that this problem is very challenging.
 
\item Use the cycle joining method to generate de Bruijn sequences from NLFSRs.

\item Find really simple rules to quickly generate de Bruijn sequences in the spirit of the {\it prefer-one method} given in~\cite{Fred82}. 
\end{enumerate}

\begin{acknowledgements}
Adamas Aqsa Fahreza wrote the \texttt{python} implementation code. The work of Z.~Chang is supported by the National Natural Science Foundation of China under Grant 61772476 and the Key Scientific Research Projects of Colleges and Universities in Henan Province under Grant 18A110029. Research Grants TL-9014101684-01 and MOE2013-T2-1-041 support the research carried out by M.~F.~Ezerman, S.~Ling, and H.~Wang.
\end{acknowledgements}


\begin{thebibliography}{10}
	\providecommand{\url}[1]{{#1}}
	\providecommand{\urlprefix}{URL}
	\expandafter\ifx\csname urlstyle\endcsname\relax
	\providecommand{\doi}[1]{DOI~\discretionary{}{}{}#1}\else
	\providecommand{\doi}{DOI~\discretionary{}{}{}\begingroup
		\urlstyle{rm}\Url}\fi
	
	\bibitem{AEB87}
	van Aardenne-Ehrenfest, T., de~Bruijn, N.G.: Circuits and trees in oriented
	linear graphs.
	\newblock Simon Stevin \textbf{28}, 203--217 (1951)
	
	\bibitem{Arn10}
	Arndt, J.: Matters Computational: Ideas, Algorithms, Source Code, 1st edn.
	\newblock Springer-Verlag, New York, NY, USA (2010)
	
	\bibitem{Bro89}
	Broder, A.: Generating random spanning trees.
	\newblock In: Proc. 30th Annual Symposium on Foundations of Computer Science, pp. 442--447 (1989)
	
	\bibitem{Bruck12}
	Bruckstein, A.M., Etzion, T., Giryes, R., Gordon, N., Holt, R.J., Shuldiner,
	D.: Simple and robust binary self-location patterns.
	\newblock IEEE Trans. Inform. Theory \textbf{58}(7), 4884--4889 (2012)
	
	\bibitem{Bruijn46}
	de~Bruijn, N.G.: A combinatorial problem.
	\newblock Koninklijke Nederlandse Akademie v. Wetenschappen \textbf{49},
	758--764 (1946)
	
	\bibitem{CELW16}
	Chang, Z., Ezerman, M.F., Ling, S., Wang, H.: The cycle structure of {LFSR}
	with arbitrary characteristic polynomial over finite fields. 
	\newblock Cryptogr. Commun. (2017). 
	\newblock \doi{10.1007/s12095-017-0273-2}, Online First 20 Dec. 2017
	
	\bibitem{Chang2017}
	Chang, Z., Ezerman, M.F., Ling, S., Wang, H.: Construction of de Bruijn
	sequences from product of two irreducible polynomials.
	\newblock Cryptogr. Commun. \textbf{10}(2), 251--275 (2018)
	
	\bibitem{Ding96}
	Ding, C., Pei, D., Salomaa, A.: Chinese Remainder Theorem: Applications in
	Computing, Coding, Cryptography.
	\newblock World Scientific Publishing Co., Inc., River Edge, NJ, USA (1996)
	
	\bibitem{EF17a}
	Ezerman, M. F., Fahreza, A. A.: A binary de Bruijn sequence generator from product of irreducible polynomials. 
	\newblock \urlprefix\url{github.com/adamasstokhorst/debruijn}

	\bibitem{EL84}
	Etzion, T., Lempel, A.: Algorithms for the generation of full-length
	shift-register sequences.
	\newblock IEEE Trans. Inform. Theory \textbf{30}(3), 480--484 (1984)
	
	\bibitem{Fred75}
	Fredricksen, H.: A class of nonlinear de {B}ruijn cycles.
	\newblock J. Combinat. Theory, Ser. A \textbf{19}(2), 192 -- 199 (1975)
	
	\bibitem{Fred82}
	Fredricksen, H.: A survey of full length nonlinear shift register cycle
	algorithms.
	\newblock SIAM Review \textbf{24}(2), 195--221 (1982)
	
	\bibitem{Golomb81}
	Golomb, S.W.: Shift Register Sequences.
	\newblock Aegean Park Press, Laguna Hills (1981)
	
	\bibitem{GG05}
	Golomb, S.W., Gong, G.: Signal Design for Good Correlation: for Wireless
	Communication, Cryptography, and Radar.
	\newblock Cambridge Univ. Press, New York (2004)
	
	\bibitem{HH96}
	Hauge, E.R., Helleseth, T.: De {B}ruijn sequences, irreducible codes and
	cyclotomy.
	\newblock Discrete Math. \textbf{159}(1-3), 143--154 (1996)
	
	\bibitem{HM96}
	Hauge, E.R., Mykkeltveit, J.: On the classification of de {B}ruijn sequences.
	\newblock Discrete Math. \textbf{148}(1–3), 65 -- 83 (1996)
	
	\bibitem{Hemmati84}
	Hemmati, F., Schilling, D.L., Eichmann, G.: Adjacencies between the cycles of a
	shift register with characteristic polynomial $(1+x)^{n}$.
	\newblock IEEE Trans. Comput. \textbf{33}(7), 675--677 (1984)
	
	\bibitem{Gray}
	Knuth, D.E.: Grayspspan.
	\newblock
	\urlprefix\url{http://www-cs-faculty.stanford.edu/~uno/programs/grayspspan.w}
	
	\bibitem{Knuth2}
	Knuth, D.E.: The Art of Computer Programming. Vol.~2 (3rd Ed.)., Seminumerical Algorithms. 
	\newblock Addison-Wesley, Longman Publishing Co., Inc., Boston (1997)
		
	\bibitem{Knuth4A}
	Knuth, D.E.: The Art of Computer Programming. Vol.~4A., Combinatorial
	Algorithms. Part 1.
	\newblock Addison-Wesley, Upple Saddle River (N.J.), London, Paris (2011)
	
	\bibitem{KSSK00}
	Kurosawa, K., Sato, F., Sakata, T., Kishimoto, W.: A relationship between
	linear complexity and k-error linear complexity.
	\newblock IEEE Trans. Inform. Theory \textbf{46}(2), 694--698 (2000)
	
	\bibitem{Li14-1}
	Li, C., Zeng, X., Helleseth, T., Li, C., Hu, L.: The properties of a class of
	linear {FSR}s and their applications to the construction of nonlinear {FSR}s.
	\newblock IEEE Trans. Inform. Theory \textbf{60}(5), 3052--3061 (2014)
	
	\bibitem{Li14-2}
	Li, C., Zeng, X., Li, C., Helleseth, T.: A class of de {B}ruijn sequences.
	\newblock IEEE Trans. Inform. Theory \textbf{60}(12), 7955--7969 (2014)
	
	\bibitem{Li16}
	Li, C., Zeng, X., Li, C., Helleseth, T., Li, M.: Construction of de {B}ruijn
	sequences from {LFSR}s with reducible characteristic polynomials.
	\newblock IEEE Trans. Inform. Theory \textbf{62}(1), 610--624 (2016)
	
	\bibitem{Li2017}
	Li, M., Jiang, Y., Lin, D.: The adjacency graphs of some feedback shift
	registers.
	\newblock Des. Codes Cryptogr. \textbf{82}(3), 695--713 (2017)
	
	\bibitem{LL17a}
	Li, M., Lin, D.: The adjacency graphs of LFSRs with primitive-like characteristic polynomials. 
	\newblock IEEE Trans. Inform. Theory \textbf{63}(2), 1325--1335 (2017).
	
	\bibitem{LL17}
	Li, M., Lin, D.: De Bruijn sequences, adjacency graphs and cyclotomy. 
	\newblock IEEE Trans. Inform. Theory \textbf{64}(4), 2941--2952 (2018)
	
	\bibitem{LN97}
	Lidl, R., Niederreiter, H.: Finite Fields.
	\newblock Encyclopaedia of Mathematics and Its Applications. Cambridge Univ.
	Press, New York (1997)
	
	\bibitem{Menezes:1996}
	Menezes, A.J., Vanstone, S.A., Oorschot, P.C.V.: Handbook of Applied
	Cryptography, 1st edn.
	\newblock CRC Press, Inc., Boca Raton, FL, USA (1996)
	
	\bibitem{M72}
	Mykkeltveit, J.: A proof of {G}olomb's conjecture for the de {B}ruijn graph.
	\newblock J. of Combinat. Theory, Ser. B \textbf{13}(1), 40 -- 45 (1972)
	
	\bibitem{NP13}
	Nagarajan, N., Pop, M.: Sequence assembly demystified.
	\newblock Nature Rev. Genet. \textbf{14}(3), 157--167 (2013)
	
	\bibitem{Ral82}
	Ralston, A.: De {B}ruijn sequences - a model example of the interaction of
	discrete mathematics and computer science.
	\newblock Math. Magazine \textbf{55}(3), 131--143 (1982)
	
	\bibitem{SG13}
	Spinsante, S., Gambi, E.: De {B}ruijn binary sequences and spread spectrum
	applications: A marriage possible?
	\newblock IEEE Trans. Aerosp. Electron. Syst. \textbf{28}(11),
	28--39 (2013)
	
\end{thebibliography}

\end{document}